\documentclass[
    pdflatex,
    sn-mathphys-num,
    hyperref={
        unicode=true,
        colorlinks=true,
        urlcolor=blue!75!black,
        linkcolor=blue!75!black,
        citecolor=blue!75!black,
        pdfusetitle,
        bookmarks=true,
        bookmarksnumbered=true,
        bookmarksopen=true
    }]{sn-jnl}
\usepackage{hyperref}
\usepackage{amsmath}
\usepackage{graphicx}
\usepackage{enumerate}
\usepackage{bbm}
\usepackage{latexsym}
\usepackage{braket}
\usepackage{multirow}
\usepackage{thm-restate}
\usepackage{amsthm,amssymb,amsthm,mathtools}
\usepackage[capitalise]{cleveref} 
\theoremstyle{plain}
\newtheorem{thm}{Theorem}
\newtheorem{lem}[thm]{Lemma}
\newtheorem{cor}[thm]{Corollary}

\newtheorem{prop}[thm]{Proposition}
\theoremstyle{definition}
\newtheorem{defn}{Definition}

\makeatletter
\if@cref@capitalise
\Crefname{thm}{Theorem}{Theorems}
\Crefname{lem}{Lemma}{Lemmas}
\Crefname{cor}{Corollary}{Corollaries}
\Crefname{conj}{Conjecture}{Conjectures}
\Crefname{defn}{Definition}{Definitions}
\Crefname{prop}{Proposition}{Propositions}
\Crefname{exam}{Example}{Examples}
\Crefname{rem}{Remark}{Remarks}
\Crefname{ques}{Question}{Questions}
\else
\crefname{thm}{theorem}{theorems}
\crefname{lem}{lemma}{lemma}
\crefname{cor}{corollary}{corollaries}
\crefname{conj}{conjecture}{conjectures}
\crefname{defn}{definition}{definitions}
\crefname{prop}{proposition}{propositions}
\crefname{exam}{example}{examples}
\crefname{rem}{remark}{remarks}
\Crefname{ques}{question}{questions}
\fi
\makeatother

\DeclareMathOperator{\Tr}{Tr}

\newcommand{\ot}{\otimes}

\newcommand{\Sym}{\mathrm{Sym}}
\newcommand{\intpart}{\mathcal C}
\newcommand{\Herm}{\mathrm{Herm}}
\newcommand{\SEP}{\mathrm{SEP}}
\newcommand{\SEPAB}{{\SEP(A\colonsep B)}}

\newcommand{\ketbra}[1]{\ket{#1}\hspace{-0.4em}\bra{#1}}
\newcommand{\dsym}{D}
\newcommand{\twirl}{\mathcal T}
\newcommand{\abs}[1]{\left|#1\right|}

\newcommand{\colonsep}{\!:\!}
\renewcommand{\rho}{\varrho}

\begin{document}

\title[]{Complete entanglement detection using polynomial invariants}


\author*[1]{\fnm{Thomas} C. \sur{Fraser}}\email{tcf@math.ku.dk}

\author[1]{\fnm{Vjosa} \sur{Blakaj}}

\author[1]{\fnm{Roberto} \sur{Rubboli}}

\author[2]{\fnm{Felix} \sur{Huber}}

\author[3]{\fnm{Marco} \sur{Fanizza}}

\affil*[1]{\orgdiv{Department of Mathematical Sciences}, \orgaddress{\street{University of Copenhagen}, \city{Copenhagen}, \postcode{2100}, \country{Denmark}}}

\affil[2]{\orgdiv{Faculty of Mathematics, Physics and Informatics}, \orgname{Institute of Informatics}, \orgaddress{University of Gdańsk, Wita Stwosza 57, \city{Gdańsk}, \postcode{80-308}, \country{Poland}}}

\affil[3]{\orgdiv{Inria}, \orgname{T\'el\'ecom Paris -- LTCI}, \orgaddress{Institut Polytechnique de Paris, \city{Palaiseau}, \country{France}}}

\abstract{
Existing methods for deciding whether a bipartite quantum state is separable or entangled typically fall into one of two categories: 
they are either complete but require access to an explicit density matrix followed by numerical optimization, or they can be evaluated directly by measuring the quantum system but are incomplete, in the sense that they cannot detect all forms of entanglement. 
In this work, we overcome both limitations in a unified framework. 
First, we bypass numerical optimization by deriving separability criteria in the form of universal bounds on tensor powers of separable states. 
We prove that these bounds are complete: every entangled state violates them for sufficiently large tensor powers. 
Second, we explicitly construct a corresponding complete family of nonlinear entanglement witnesses, which can detect all forms of entanglement without requiring an explicit density matrix. 
The witnesses we construct are moreover basis-independent, in the sense that they are invariant under conjugation by local unitaries. 
Altogether, our results expand the toolbox for entanglement detection in arbitrary local dimensions in a manifestly invariant way.
}


\maketitle

\section{Introduction}

Determining whether a quantum state $\rho_{AB}$ is entangled or separable is a fundamental task in quantum theory.
Despite being an intrinsically challenging task, from both complexity-theoretic \cite{gurvits2004classical} and geometric perspectives \cite{aubrun2015dvoretzky},
today there exist a myriad of techniques for certifying entanglement in bipartite systems.
When an explicit density matrix $\rho_{AB}$ is available, entanglement can be certified, for example, using the PPT criterion \cite{peres1996separability} or $k$-symmetric-extension criteria \cite{doherty2004complete}. They give rise to complete hierarchies of feasibility problems that can, in principle, be tackled using semidefinite programming \cite{doherty2004complete,navascues2009complete,harrow2017improved,pena2025tailored}.
On the other hand, when provided with a physical quantum state shared between two subsystems, as in an experimental setting, approaches based on linear and nonlinear entanglement witnesses are more suitable \cite{guhne2006nonlinear,chruscinski2014entanglement,rico2024entanglement,elben2020mixed,neven2021symmetry,miller2026detecting,gulati2024using}.

Throughout this paper, we denote by $\SEPAB$ the set of separable states on a bipartite Hilbert space $A \otimes B$ of finite dimension $d_{AB} = \dim(A \otimes B)$.
A \textit{(polynomial) entanglement witness} is a polynomial $p(\rho_{AB})$ of degree $\ell \in \mathbb N$ specified by a Hermitian observable $W_{A^{\ell}B^{\ell}}$ on $\ell$-copies of the bipartite quantum system, whose expectation value
\begin{equation}
    p(\rho_{AB}) = \Tr(W_{A^{\ell}B^{\ell}} \rho_{AB}^{\otimes \ell}),
\end{equation}
is a polynomial in $\rho_{AB}$ that is non-negative on all separable states:
\begin{equation}
    \forall \rho_{AB} \in \SEPAB \implies p(\rho_{AB}) \geq 0.
\end{equation}
If $p(\rho_{AB}) < 0$ holds for some bipartite state $\rho_{AB}$, then $\rho_{AB}$ is entangled.
In addition to being naturally more probative than linear entanglement witnesses, a polynomial entanglement witness can have the additional desirable property of being local-unitary invariant, i.e., for all unitaries $U_{A}$ on $A$ and $V_{B}$ on $B$,
\begin{equation}
    p(\rho_{AB}) = p((U_{A} \otimes V_{B}) \rho_{AB} (U_{A}^{\dagger} \otimes V_{B}^{\dagger} )).
\end{equation} 
Since the entanglement of a state $\rho_{AB}$ is unchanged under local reversible operations, the property of being local-unitary invariant is a minimal prerequisite for any sensible measure of entanglement \cite{grassl1998computing, leifer2004measuring,schmid2023understanding}.

A set of entanglement witnesses, $\mathcal P$, is said to be \textit{complete} for local dimensions $(d_{A},d_{B})$ if for every entangled state $\rho_{AB}$, there exists some witness $p \in \mathcal P$, which witnesses the entanglement of $\rho_{AB}$, i.e., $p(\rho_{AB}) < 0$.
For example, for local dimensions $(d_{A}, d_{B})$ where $d_{A}d_{B} \leq 6$, the PPT criterion \cite{peres1996separability} is both a necessary and sufficient criteria for separability \cite{horodecki2001separability}, and has been converted into a finite collection of polynomial entanglement witnesses that is \textit{complete} and experimentally accessible \cite{elben2020mixed,neven2021symmetry,miller2026detecting}.

Finding a complete collection of entanglement witnesses for quantum systems of larger dimensions has remained an open challenge.
In this work, we address this fundamental limitation by developing a complete family of entanglement witnesses for any finite local dimensions $(d_{A}, d_{B})$.  

\subsection{Contributions}
\label{sec:contributions}

This paper provides characterizations of the set of separable states, formulated both as a family of operator bounds on tensor powers and as a family of polynomial entanglement witnesses.
These characterizations are complete, in the sense that every entangled state violates at least one bound in the family.

Our central contribution is to identify a specific probability measure $\mu$ on the set of separable states $\SEPAB$ such that, for each positive integer $n$, the corresponding de Finetti state
\begin{equation}
    \Omega_{A^{n}B^{n}} = \int_{\SEP} \!\!d \mu(\sigma_{AB}) \sigma_{AB}^{\otimes n}
\end{equation}
serves as a universal upper bound on tensor powers of separable states. More precisely, for every
$\rho_{AB} \in \SEPAB$, the operator $\rho_{AB}^{\otimes n}$ is bounded from above by $\Omega_{A^{n}B^{n}}$ up to a prefactor $f_{AB}(n)$, which grows polynomially in $n$ for fixed local dimensions $d_A$ and $d_B$.
Throughout this paper, operator inequalities are understood in the Loewner order, i.e., $X\preceq Y$ means the difference $(Y-X)$ is positive semidefinite.

\begin{restatable}{thm}{forward}
\label{thm:forward}
    If $\rho_{AB} \in\SEPAB$, then for all $n \in \mathbb{N}$,
    \begin{equation}
        \label{eq:forward}
        \rho_{AB}^{\otimes n} \preceq f_{AB}(n)\Omega_{A^{n} B^{n}},
    \end{equation}
    where $f_{AB}(n)$ is a function of $n \in \mathbb{N}$ satisfying the bound $f_{AB}(n) \leq (n+1)^{d_{AB}^3 - 1}$.
\end{restatable}

The immediate application of \cref{thm:forward} is to certify the entanglement of a state $\rho_{AB}$ --
if a state $\rho_{AB}$ happens to violate \cref{eq:forward} for some value of $n$, then \cref{thm:forward} implies that $\rho_{AB}$ is entangled.
For this purpose, we prove a slightly stronger statement than \cref{thm:forward}. Namely, we identify an intermediate operator $\Lambda_{A^{n}B^{n}}$ such that, for all $\rho_{AB} \in\SEPAB$, $\rho_{AB}^{\otimes n} \preceq \Lambda_{A^{n}B^{n}} \preceq f_{AB}(n)\Omega_{A^{n} B^{n}}$. We refer to \cref{thm:stronger_forward} for the precise statement. 

Another application of \cref{thm:forward} is a quantitative bound on the probabilities of \textit{arbitrary} measurements performed on $n$-copies of separable states. 
In particular, \cref{thm:forward} implies the following corollary, which states that the maximum likelihood of any measurement outcome on $n$-copies of a separable state cannot deviate too much, from the average likelihood taken with respect to the measure $\mu$.
\begin{cor}
    \label{cor:avg_mean_likelihood}
    Let $0 \preceq M_{A^{n}B^{n}} \preceq I_{AB}^{\otimes n}$ be a measurement effect.
    Then,
    \begin{equation}
        \Tr[M_{A^{n}B^{n}} \Omega_{A^{n}B^{n}}] \leq \max_{\sigma_{AB} \in \SEP} \Tr[M_{A^{n}B^{n}} \sigma_{AB}^{\otimes n}] \leq f_{AB}(n) \Tr[M_{A^{n}B^{n}} \Omega_{A^{n}B^{n}}].
    \end{equation}
\end{cor}
From the perspective of entanglement detection, \cref{cor:avg_mean_likelihood} implies that if a state $\rho_{AB}$ achieves a likelihood $\Tr[M_{A^{n}B^{n}} \rho_{AB}^{\otimes n}]$ much greater than the average likelihood $\Tr[M_{A^{n}B^{n}} \Omega_{A^{n}B^{n}}]$ over all separable states, then $\rho_{AB}$ must be entangled.
Note that the first inequality in \cref{cor:avg_mean_likelihood} holds trivially while the second inequality is a direct consequence of \cref{thm:forward}.

Using this operational perspective, together with techniques from quantum state tomography, we also prove the \textit{converse} of \cref{thm:forward}.
In other words, if a state $\rho_{AB}$ satisfies $\rho_{AB}^{\otimes n} \preceq f_{AB}(n)\Omega_{A^{n} B^{n}}$ for every $n \in \mathbb N$, then one can conclude that $\rho_{AB}$ is indeed separable.
In fact, we formulate the converse statement in a quantitative manner using the \textit{logarithmic fidelity of separability} \cite{streltsov2010linking}
\begin{equation}
    G(\rho_{AB}) \coloneqq - \log F_{\mathrm{s}}(\rho_{AB}), \quad \text{with} \quad F_{\mathrm{s}}(\rho_{AB})\coloneqq \max_{\sigma_{AB} \in \SEP} F(\rho_{AB}, \sigma_{AB}),
\end{equation}
where $F(\rho, \sigma) \coloneqq \Tr(\abs{\sqrt{\rho}\sqrt{\sigma}})^2$ is the (squared) fidelity between $\rho$ and $\sigma$. 
The quantity $G(\rho_{AB})$ is faithful, in the sense that $G(\rho_{AB}) \geq 0$ for all states $\rho_{AB}$, and $G(\rho_{AB}) = 0$ if and only if $\rho_{AB} \in \SEPAB$. We
now state our quantitative converse to \cref{thm:forward}.

\begin{restatable}{thm}{completeness}
    \label{thm:completeness}
    Let $n \in \mathbb N$ be such that $n > d_{AB}^{2}$ and suppose that $\rho_{AB}$ satisfies $\rho_{AB}^{\otimes n} \preceq f_{AB}(n) \Omega_{A^{n} B^{n}}$.
    Then
    \begin{equation}
        \label{eq:G_bound}
        G(\rho_{AB}) \leq  \frac{\log (C(\rho_{AB})^{-1} f_{AB}(n))}{n-d_{AB}^{2}},
    \end{equation}
    where $C(\rho_{AB}) \coloneqq x_1^{s} x_2^{s-1} \cdots x_{s-1}^{2} x_{s} \in (0,1]$
    with $x_1 \geq \cdots \geq x_{s} > 0$ being the non-zero eigenvalues of $\rho_{AB}$, and $s = \mathrm{rank}(\rho_{AB})$.
\end{restatable}
This result shows that the operator inequalities in \cref{thm:forward}, taken for all $n \in \mathbb N$, provide a \textit{complete} characterization of $\SEP(A\colonsep B)$.
Indeed, since $f_{AB}(n)$ grows only polynomially in $n$, the upper bound in \cref{eq:G_bound} scales like $O((\log n)/n)$, and therefore vanishes as $n \to \infty$.
Hence, if $\rho_{AB}^{\otimes n} \preceq f_{AB}(n) \Omega_{A^{n} B^{n}}$ holds for all arbitrarily large $n$, we conclude $G(\rho_{AB}) = 0$ and therefore $\rho_{AB} \in \SEPAB$.
In other words, \cref{thm:completeness} states that every entangled state $\rho_{AB}$ must violate the criterion in \cref{thm:forward}, that is, $\rho_{AB}^{\otimes n} \not \preceq f_{AB}(n)\Omega_{A^nB^n}$ for all sufficiently large $n$, where \cref{eq:G_bound} fails.

Building on this complete characterization of $\SEP(A\colonsep B)$ in terms of operator inequalities, we additionally construct a countable family of polynomial entanglement witnesses that is complete.
These entanglement witnesses arise from a technical lemma, \cref{lem:witness_from_lambda}, based on Sylvester's criterion for operator positivity, which allows us to convert operator inequalities into polynomial inequalities. Using this lemma, we construct polynomial entanglement witnesses indexed by $n \in \mathbb N$ and $m \in \{1,\ldots, d_{AB}^{n}\}$, given by 
\begin{equation}
    p_{n,m}(\rho_{AB}) = \Tr[P_{\wedge^{m}(A^nB^n)} (f_{AB}(n)\Omega_{A^{n} B^{n}}-\rho_{AB}^{\otimes n})^{\otimes m}],
\end{equation}
where $P_{\wedge^{m}(A^nB^n)}$ denotes the projection onto the antisymmetric subspace $\wedge^{m}(A^nB^n)$ of $m$-tensor copies of $(A\otimes B)^{\otimes n}$.
We additionally show that $p_{n,m}(\rho_{AB})$ can be identified with a Hermitian observable, $W_{A^{nm}B^{nm}}^{n,m}$, which is \textit{independent} of the state $\rho_{AB}$ being tested. More precisely, for every states $\rho_{AB}$,
\begin{equation}
    p_{n,m}(\rho_{AB})=\Tr(W_{A^{nm}B^{nm}}^{n,m} \rho_{AB}^{\otimes nm}).
\end{equation}
An explicit formula for $W_{A^{nm}B^{nm}}^{n,m}$, in terms of $f_{AB}(n)\Omega_{A^{n} B^{n}}$ and $m$, is given in \cref{lem:witness_from_lambda}.
Moreover, we show that these polynomial entanglement witnesses are local-unitary invariant, thereby respecting a fundamental symmetry of the set of separable states.

Using these entanglement witnesses, we obtain the following complete characterization of $\SEPAB$:
\begin{cor}
    \label{cor:complete_witnesses}
    A state $\rho_{AB}$ is separable if and only if, for all $n \in \mathbb N$ and all $m \in \{1,\ldots, d_{AB}^{n}\}$,
    \begin{equation}
        p_{n,m}(\rho_{AB}) \geq 0.
    \end{equation}
\end{cor}
The advantage of \cref{cor:complete_witnesses} over \cref{thm:stronger_forward} is that the evaluation of $p_{n,m}(\rho_{AB})=\Tr(W_{A^{nm}B^{nm}}^{n,m} \rho_{AB}^{\otimes nm})$ does not depend on having a classical description of $\rho_{AB}$. 
In principle, $p_{n,m}(\rho_{AB})$ may be estimated by direct measurement on $nm$-copies of an unknown state $\rho_{AB}$.

In the subsequent sections, we develop the techniques required to prove these results.
\Cref{sec:preliminaries} begins by reviewing separable and entangled bipartite states in \cref{sec:entanglement}, the non-commutative multinomial expansion for tensor powers of operator sums in \cref{sec:multinomial_expansions}, and the theory of symmetric and block-wise symmetric subspaces in \cref{sec:symmetric_subspaces}.
Our proof of \cref{thm:completeness} also relies on a likelihood-ratio and fidelity bound between quantum states, which we review in \cref{sec:likelihood_fidelity}.

With these prerequisites established, \cref{sec:results} provides the proofs of our main contributions. First, we introduce the aforementioned probability measure $\mu$ in \cref{sec:random_sep_states} and derive a formula for the de Finetti operator $\Omega_{A^{n}B^{n}}$ associated with this measure. Second, we prove a stronger version of \cref{thm:forward}, namely \cref{thm:stronger_forward}, in \cref{sec:necessary_inequalities}. Third, we use the likelihood-ratio and fidelity bound from \cref{sec:likelihood_fidelity} to prove \cref{thm:completeness}.
Fourth, in \cref{sec:poly_ent_witnesses}, we show how to construct a complete family of local-unitary invariant polynomial entanglement witnesses, thereby establishing \cref{cor:complete_witnesses}.
Finally, \cref{sec:discussion} concludes with a discussion and outlook. 

\section{Preliminaries}
\label{sec:preliminaries}
\subsection{Entanglement}
\label{sec:entanglement}
Given a $d_{H}$-dimensional complex Hilbert space $H \simeq \mathbb C^{d_{H}}$, we let the set of density matrices on $H$ be denoted by
\begin{equation}
    \mathcal S(H) \coloneqq \{ \rho : H \to H \mid \rho = \rho^{\dagger}, \rho \succeq 0, \Tr(\rho) = 1  \}.
\end{equation}
\begin{defn}
    \label{defn:separable}
    A bipartite quantum state $\rho_{AB} \in \mathcal S(A\otimes B)$ is \textit{separable} if it can be decomposed into the following form:
    \begin{equation}
        \label{eq:sep_r_ab}
        \rho_{AB} = \sum_{i=1}^{r} p_i \ketbra{\alpha_i}_{A} \otimes \ketbra{\beta_i}_{B},
    \end{equation}
    where $r$ is a positive integer, $(p_1,\ldots, p_r)$ is a probability distribution, and $\ket{\alpha_i}$ and $\ket{\beta_i}$ are pure states in the local Hilbert spaces $A$ and $B$.
    If $\rho_{AB}$ cannot be expressed in this form, then $\rho_{AB}$ is said to be \textit{entangled}.
    We denote the set of all separable states by $\SEP(A\colonsep B)$.
\end{defn}

The \textit{minimum} value of $r$ needed to express a given separable state $\rho_{AB}$ as in \cref{eq:sep_r_ab} is known as the \textit{separable rank} of $\rho_{AB}$.\footnote{Note that the separable rank is distinct from the notion of \textit{mixed} separable rank, where the local states $\ketbra{\alpha_i}_{A}$ and $\ketbra{\beta_i}_{B}$ are permitted to be mixed states $\sigma_{A}$ and $\sigma_{B}$ \cite{gribling2022bounding}.}
By definition, the set of \textit{all} separable states $\SEPAB$  allows decompositions with arbitrary rank $r$. Nevertheless, by Carathéodory's theorem, every separable state admits a decomposition with separable rank at most 
\cite{gribling2022bounding}
\begin{equation}
    \label{eq:sep_rank_finite}
    r \leq \mathrm{rank}(\rho_{AB})^2 \leq d_{AB}^{2}.
\end{equation}
Throughout this paper, when working with separable states, we will make frequent use of this finite and dimension dependent bound on the separable rank $r$.
In particular, without loss of generality, we always assume that each separable state $\rho_{AB}$ has a decomposition of the form \cref{eq:sep_r_ab}, with $r = d_{AB}^{2}$.

\subsection{Multinomial expansions}
\label{sec:multinomial_expansions}
A central ingredient in our techniques is to expand tensor powers of separable states, which are convex combinations of product states. 
We therefore begin by describing, in a more general setting, how to decompose tensor powers of convex combinations of quantum states\footnote{Tensor-power multinomial expansions of this form have also been used to derive various bounds on sums of positive semidefinite matrices in \cite{tie2011rearrangement,lin2014completely,alekseev2025tetrahedral}.}.
As an example, let $H$ be a finite-dimensional Hilbert space, let $p \in [0,1]$ and consider an operator of the form $\rho = p \sigma + (1-p) \eta$, where $\sigma$ and $\eta$ are operators on $H$. 
Then, $\rho^{\otimes 3}$ can be expanded as
\begin{align}
    \begin{split}
        \rho^{\otimes 3} 
        &= \big(p \sigma + (1-p) \eta\big)^{\otimes 3}\\
        &= p^3 (\sigma \ot \sigma \ot \sigma) + (1-p)^3 (\eta \ot \eta \ot \eta) \\
        &\quad+ p^{2}(1-p) \left(\sigma^{\otimes 2} \ot \eta + \sigma \ot \eta \ot \sigma + \eta \ot \sigma^{\otimes 2}\right) \\
        &\quad+ p(1-p)^{2} \left(\sigma \ot \eta^{\otimes 2} + \eta \ot \sigma \ot \eta + \eta^{\otimes 2} \ot \sigma\right).
    \end{split}
\end{align}
Many terms in the above expansion are related by permutations of the tensor factors. For notational convenience, we therefore group together all terms that can be obtained from one another by \textit{permuting} the tensor factors.
Let $T : S_{n} \to \mathrm{GL}(H^{\otimes n})$ denote the \textit{tensor-permutation representation} of the symmetric group $S_{n}$ on $H^{\otimes n}$, defined on product vectors via
\begin{align}
    \label{eq:tensor_perm}
    T(\pi) \ket {v_1} \ot \ket{v_2} \ot \cdots \ot \ket{v_n}
    &= \ket {v_{\pi^{-1}(1)}} \ot \ket{v_{\pi^{-1}(2)}} \ot \cdots \ot \ket{v_{\pi^{-1}(n)}}.
\end{align}
We further define the symmetric group \textit{twirling} 
operator $\twirl_{S_n}$, acting on operators $X$ on $H^{\otimes n}$ by conjugation
\begin{equation}
    \label{eq:twirl}
    \twirl_{S_n}(X) = \frac{1}{n!} \sum_{\pi \in S_n} T(\pi) X T(\pi^{-1}).
\end{equation}

For example, the action of twirling the operator $\sigma^{\otimes 2} \ot \eta$ by permutations $\pi \in S_3$ is given by
\begin{align}
    \begin{split}
        \twirl_{S_3}(\sigma^{\otimes 2} \ot \eta)
        &= \frac{1}{3!} \sum_{\pi \in S_3} T(\pi) (\sigma^{\otimes 2} \ot \eta) T(\pi^{-1}) \\ 
        &= \frac{2!1!}{3!}\left( \sigma^{\otimes 2} \ot \eta + \sigma \ot \eta \ot \sigma + \eta \ot \sigma^{\otimes 2} \right),
    \end{split}
\end{align}
where the coefficient $2!1! = 2$ accounts for the size of the subgroup that leaves $\sigma^{\otimes 2} \ot \eta$ invariant.
To state the general result, we first introduce some notation.
\begin{defn}
    For positive integers $n$ and $r$, let $\intpart(n,r)$ denote the set of \textit{weak compositions} of $n$ into $r$ parts, i.e.,
    \begin{equation}
        \label{eq:intpart}
        \intpart(n,r) \coloneqq \{ (k_1,\ldots, k_r) \in \mathbb Z_{\geq 0}^{r} \mid k_1 + \cdots + k_r = n\}.
    \end{equation}
\end{defn}
A weak composition $\kappa = (k_1,\ldots, k_r)$ of $n$ into $r$ parts can also be viewed as an unsorted integer partition of $n$ into $r$ (possibly empty) parts.
\begin{defn}
    Given a probability distribution $p = (p_1,\ldots, p_r)$ and a weak composition $\kappa = (k_1,\ldots,k_r) \in \intpart(n,r)$, we let the \textit{multinomial distribution} be denoted by
    \begin{equation}
        M_{p}(\kappa) = \binom{n}{k_1,\ldots,k_r} p_1^{k_1}\cdots p_r^{k_r}.
    \end{equation}
\end{defn}
With this notation in mind, we can now describe more generally how tensor powers of convex combinations of states can be decomposed.
\begin{prop}
    \label{prop:multinomial_expansion}
    Let $\rho$ be a convex combination of $\sigma_1,\ldots,\sigma_r$ with convex weights $p = (p_1,\ldots, p_r)$
    \begin{equation}
        \label{eq:decomposition}
        \rho = \sum_{i=1}^{r} p_i \sigma_{i}.
    \end{equation}
    Then, $\rho^{\ot n}$ admits the expansion 
    \begin{align}
        \label{eq:multinomial_expansion}
        \rho^{\ot n}
        = \sum_{\kappa \in \intpart(n,r)} M_{p}(\kappa) \twirl_{S_n}(\sigma_{1}^{\otimes k_1} \ot \cdots \ot \sigma_{r}^{\otimes k_r}).
    \end{align}
\end{prop}
In subsequent sections, we apply \cref{prop:multinomial_expansion} to separable states expressed as convex decompositions of pure product states.

\subsection{Symmetric subspaces}
\label{sec:symmetric_subspaces}

Our results also rely heavily on the structure of the symmetric subspaces $\Sym^{n}(H) \subseteq H^{\otimes n}$ of tensor powers of $d_{H}$-dimensional Hilbert spaces.
The symmetric subspace, denoted $\Sym^{n}(H)$, is the subspace of $S_n$-invariant vectors in $H^{\otimes n}$
\begin{equation}
    \Sym^{n}(H) = \{ \ket{\Phi} \in H^{\otimes n} \mid \forall \pi \in S_n : T(\pi)\ket{\Psi} = \ket{\Psi}\}.
\end{equation}
For a $d_{H}$-dimensional Hilbert space $H \simeq \mathbb C^{d_{H}}$, there exists a unique $U(d_{H})$-invariant probability measure, $d \psi$, on the set of pure states $\ket{\psi} \in H$, induced by the Haar measure on $U(d_{H})$ \cite{harrow2013church}.
With respect to this measure, the mean value of the the $n$-th tensor power $\ketbra{\psi}^{\otimes n}$ is given by 
\begin{equation}
    \label{eq:integral_sym_proj}
    \int d\psi \ketbra{\psi}^{\ot n} = \frac{\Pi_{\Sym^{n}(H)}}{\dsym_H(n)},
\end{equation}
where $\Pi_{\Sym^{n}(H)}$ denotes the projection operator onto the symmetric subspace $\Sym^{n}(H) \subseteq H^{\otimes n}$ \cite{harrow2013church}.
The quantity $\dsym_H(n)$ denotes the dimension of the symmetric subspace:
\begin{equation}
    \dsym_H(n) \coloneqq \Tr(\Pi_{\Sym^{n}(H)}) = \dim \Sym^{n}(H) = \binom{n+d_H-1}{n}.
\end{equation}

In \cref{sec:multinomial_expansions} we introduced the set $\intpart(n,r)$ of weak compositions of $n$ into $r$ parts, as indexing outcomes $\kappa = (k_1,\ldots, k_r)$ for the multinomial distribution $M_p(\kappa)$.
Here, we remark that $\kappa \in \intpart(n,r)$ also indexes an orthonormal basis for $\Sym^{n}(H)$, called the \textit{Dicke basis} (whenever $d_{H} = r$). 
Given a fixed orthonormal set of vectors $\{\ket 1, \ket 2, \ldots, \ket{r}\}$ and $\kappa = (k_1,\ldots,k_r) \in \intpart(n,r)$, we let $\ket{\Psi_{\kappa}}$ denote the vector
\begin{equation}
    \ket{\Psi_{\kappa}} \coloneqq \ket 1^{\otimes k_1} \otimes \cdots \otimes \ket r^{\otimes k_r}.
\end{equation}
Note that $\ket{\Psi_{\kappa}}$ generally does not belong to the symmetric subspace $\Sym^{n}(H)$ but rather in the larger \textit{block-wise} symmetric subspace, denoted by\footnote{Throughout this paper, we use the notational convention that $\Sym^{0}(H) \coloneqq \mathbb C$ for a Hilbert space $H$.}
\begin{equation}
    \Sym^{\kappa}(H) \coloneqq \Sym^{k_1}(H) \otimes \cdots \otimes \Sym^{k_r}(H).
\end{equation}
By projecting $\ket{\Psi_{\kappa}}$ into the symmetric subspace and re-normalizing, one obtains the Dicke basis state associated to $\kappa$:
\begin{equation}
    \ket{\dsym_{\kappa}} \coloneqq \sqrt{\binom{n}{k_1,\ldots,k_r}} \Pi_{\Sym^{n}(H)}\ket{\Psi_{\kappa}}.
\end{equation}
The set of Dicke states $\{\ket{\dsym_{\kappa}} \mid \kappa \in  \intpart(n,r)\}$ forms an orthonormal basis for $\Sym^{n}(H)$ when the Hilbert space $H$ has dimension $d_{H} = r$.

In subsequently sections, we will need to bound the dimension of the block-wise symmetric subspace,
\begin{equation}
    \dsym_{H}(\kappa) = \dim \Sym^{\kappa}(H). 
\end{equation}
by a quantity that depends only on the total $n = k_1 +\cdots + k_r$ for a fixed number of parts $r$ and the dimension $d_{H} = \dim H$.
In particular, we observe that $\dsym_{H}(\kappa)$ is maximized when the parts $k_i$ are as uniform as possible, that is, when $k_i \approx n/r$.
\begin{lem}
    \label{lem:blockwise_bound}
    Let $\kappa = (k_1, \ldots, k_r) \in \intpart(n,r)$. Then,
    \begin{equation}
        \dsym_{H}(\kappa) \leq \dsym_{H}\left(\left\lceil\frac{n}{r}\right\rceil\right)^{r}.
    \end{equation}
\end{lem}
\begin{proof}
    The proof follows by noting that the function $k \mapsto \dsym_{H}(k)$ is a log concave function for any dimension $d_{H} = \dim H$.
    Indeed, for any $k$:
    \begin{equation}
        \frac{\dsym_{H}(k+1)}{\dsym_{H}(k)} = \frac{\binom{k+1+d_H-1}{d_H-1}}{\binom{k+d_H-1}{d_H-1}} = \frac{k+d_{H}}{k+1} = 1 + \frac{d_{H}-1}{k+1} \leq \frac{\dsym_{H}(k)}{\dsym_{H}(k-1)},
    \end{equation}
    and therefore
    \begin{equation}
        \dsym_{H}(k)^2 \geq \dsym_{H}(k+1)\dsym_{H}(k-1).
    \end{equation}
    By iterating this inequality one sees that the dimension of the block-wise symmetric subspace, $\dsym_{H}(\kappa)$, can only be maximized when the components of $\kappa = (k_1, k_2, \ldots, k_r)$ differ by at most one, i.e., $\abs{k_i - k_j} \leq 1$.
    In turn, this implies that $\dsym_{H}(\kappa)$ is maximized when each block has size $k_i \in \{\lfloor\frac{n}{r}\rfloor, \lceil\frac{n}{r}\rceil \}$, and therefore
    \begin{equation}
        \max_{\kappa \in \intpart(n,r)} \dsym_{H}(\kappa) \leq \dsym_{H}\left(\left\lceil\frac{n}{r}\right\rceil\right)^{r}.
    \end{equation}
\end{proof}

\subsection{Likelihood-ratios and state fidelity}
\label{sec:likelihood_fidelity}

Essential to the proof of \cref{thm:completeness} is the idea that one can always distinguish $n$-copies of a state $\rho$ from $n$-copies of any different state $\sigma \neq \rho$ by performing suitable measurements.  
Specifically, there must exist sequences of measurement effects $M_n \succeq 0$ such that the probability $\Tr(M_n \rho^{\otimes n})$ is large, while the probability $\Tr(M_n \sigma^{\otimes n})$ is small whenever $\sigma \neq \rho$.
To quantify this distinction, we consider the associated likelihood-ratio.
\begin{defn}
    Given $\rho, \sigma \in \mathcal S(H)$ and measurement effect $0\preceq M_n \preceq I_{H}^{\otimes n}$ acting on $H^{\otimes n}$, we define the \textit{likelihood-ratio} by
    \begin{equation}
        \mathcal R_{\sigma/\rho}(M_n) \coloneqq \frac{\Tr(M_n \sigma^{\otimes n})}{\Tr(M_n \rho^{\otimes n})}.
    \end{equation}
\end{defn}
To quantify the difference between two quantum states $\rho, \sigma \in \mathcal S(H)$, we use the quantum state fidelity.
\begin{defn}
    The fidelity between two quantum states $\rho, \sigma \in \mathcal S(H)$  is defined by
    \begin{equation}
        F(\rho, \sigma) \coloneqq \Tr(\abs{\sqrt{\rho}\sqrt{\sigma}})^2.
    \end{equation}
    In general, $0 \leq F(\rho, \sigma) \leq 1$ with
    \begin{equation}
        \label{eq:faithful_fidelity}
        F(\rho, \sigma) = 1 \quad \Longleftrightarrow \quad \rho = \sigma.
    \end{equation}
\end{defn}
For our purposes, we will use the following likelihood-ratio bound in terms of the fidelity between $\rho$ and $\sigma$ \cite{fraser2022sufficient}.
\begin{lem}
    \label{lem:likelihood_ratio_fidelity}
    Let $\rho \in \mathcal S(H)$ be a $d_{H}$-dimensional quantum state.
    Then, there exists a sequence of measurements $M_n \succeq 0$ such that for all $n \in \mathbb N$ and $\sigma \in \mathcal S(H)$
    \begin{equation}
         \mathcal R_{\sigma/\rho}(M_n) = \frac{\Tr(M_n \sigma^{\otimes n})}{\Tr(M_n \rho^{\otimes n})}\leq C(\rho)^{-1} F(\rho, \sigma)^{n-d_{H}^{2}},
    \end{equation}
    where $C(\rho) \in (0,1]$ is given by 
    \begin{equation}
        \label{eq:staircase_constant}
        C(\rho) \coloneqq x_1^{s} x_2^{s-1} \cdots x_{s-1}^{2} x_{s},
    \end{equation}
    with $(x_1 \geq \cdots \geq x_{d_{H}})$ being the eigenvalues of $\rho$, and $s = \mathrm{rank}(\rho)$.
\end{lem}
\begin{proof}
    An explicit construction for the sequence of measurements $M_n \succeq 0$ which satisfies the claim was given in \cite[Appendix J, Theorem 15]{fraser2022sufficient}. 
    The proof is also reproduced in \cref{sec:state_estimation}, with minor improvements.
\end{proof}
The statement of \cref{lem:likelihood_ratio_fidelity} admits the following interpretation. If $\sigma=\rho$, then $\mathcal R_{\rho/\rho}(M_n)=1$ and $F(\rho,\rho)=1$, so the bound is trivial because $C(\rho)^{-1}\geq 1$.
On the other hand, if $\sigma\neq\rho$, then $F(\rho,\sigma)<1$, and \cref{lem:likelihood_ratio_fidelity} shows that the likelihood ratio $\mathcal R_{\sigma/\rho}(M_n)$ decays to zero as $n\to\infty$ at an exponential rate determined by $-\log F(\rho,\sigma)$. Beyond this exponential decay, the essential feature of the \cref{lem:likelihood_ratio_fidelity} is that the measurements $M_n$ are \textit{independent} of $\sigma$.

\section{Complete entanglement detection}
\label{sec:results}

\subsection{Random separable states}
\label{sec:random_sep_states}

This section introduces a probability measure $\mu$ on the set of separable states $\SEPAB$. We later obtain an analytic formula for the associated de Finetti state $\Omega_{A^{n}B^{n}}$; see \cref{lem:sep_de_finetti_formula}. The state is given by
\begin{align}
    \label{eq:de_finetti_state}
    \begin{split}
    \Omega_{A^{n}B^{n}}
    &= \int d \mu(\sigma_{AB}) \sigma_{AB}^{\ot n}.
    \end{split}
\end{align}
To define the measure $\mu$, we introduce an auxiliary register Hilbert space $R$ with dimension 
\begin{equation}
    r = d_{R} = d_{AB}^{2}.
\end{equation}
The purpose of this register is to index the summands in the separable decomposition of a generic separable state $\rho_{AB}$.

Now let $\ket{\alpha_1}_{A},\ldots, \ket{\alpha_r}_{A} \in A$, $\ket{\beta_1}_{B},\ldots, \ket{\beta_r}_{B} \in B$, and $\ket{\gamma}_{R} \in R$ be independent uniformly random pure states in their respective Hilbert spaces. 
Using these random pure states, the measure $\mu$ is induced by constructing the following separable state $\sigma_{AB}$ according to the formula
\begin{equation}
    \label{eq:random_sep}
    \sigma_{AB} = \sum_{i=1}^{r} \abs{\braket{i|\gamma}_{R}}^2 \ketbra{\alpha_i}_{A} \ot \ketbra{\beta_i}_{B},
\end{equation}
where the weights $p_i = \abs{\braket{i|\gamma}_{R}}^2$ are given by the squared coordinates of the random pure state $\ket{\gamma} \in R$ taken with respect to any orthonormal basis $\{\ket 1_{R}, \ldots, \ket{r}_{R}\}$ for $R = \mathbb C^{r}$.
This choice of weights is equivalent to letting the distribution $p = (p_1,\ldots, p_{r})$ be distributed according to the uniform Dirichlet measure \cite{zyczkowski2001induced}.

As we shall demonstrate momentarily, this construction gives rise to a tractable formula for the average $n$-th tensor power of $\sigma_{AB}$ distributed according to $\mu$, analogous to \cref{eq:integral_sym_proj}.
To state our result, first recall from the multinomial expansion formula (\cref{prop:multinomial_expansion}), that $\intpart(n,r)$ denotes the set of weak compositions $\kappa = (k_1, \ldots, k_r)$ of $n$ (\cref{eq:intpart}), and $\twirl_{S_n}$ denotes the symmetric group twirling operator (\cref{eq:twirl}).
\begin{lem}
    \label{lem:sep_de_finetti_formula}
    The de Finetti state $\Omega_{A^{n}B^{n}}$ in \cref{eq:de_finetti_state} 
    admits the representation 
    \begin{equation}
        \label{eq:omega_formula}
        \Omega_{A^{n}B^{n}}
        = \frac{1}{\dsym_R(n)}\sum_{\kappa \in \intpart(n,r)} \twirl_{S_n}\left(\frac{\Pi_{\Sym^{\kappa}(A)}}{\dsym_{A}(\kappa)} \ot \frac{\Pi_{\Sym^{\kappa}(B)}}{\dsym_{B}(\kappa)}\right),
    \end{equation}
    where $\dsym_H(\kappa)$ is the dimension of the block-wise symmetric subspace $\Sym^{\kappa}(H)$ of a Hilbert space $H$ and $r=d_{R} = d_{AB}^{2}$.
\end{lem}
\begin{proof}
    First apply the multinomial expansion from \cref{prop:multinomial_expansion} to the integrand in \cref{eq:de_finetti_state}, with $p_i = \abs{\braket{i|\gamma}_{R}}^2$ and $\sigma_i = \ketbra{\alpha_i}_{A} \ot \ketbra{\beta_i}_{B}$, to obtain
    \begin{equation}
        \Omega_{A^{n}B^{n}}
        = \int d\gamma \int \prod_{i=1}^{r} d \alpha_i d \beta_i  \sum_{\kappa \in \intpart(n,r)} M_{p}(\kappa) \twirl_{S_n}
        \left(\bigotimes_{i=1}^{r}\ketbra{\alpha_i}_A^{\otimes k_i} \ot \ketbra{\beta_i}_B^{\otimes k_i} \right). \\
    \end{equation}
    The claimed formula follows by evaluating this integral through repeated applications of \cref{eq:integral_sym_proj}.
    The resulting expression involves integrals over the pure states $\ket{\gamma}_{R}$, $\ket{\alpha_i}_{A}$, and $\ket{\beta_i}_{B}$, which we address separately.
    For each $\kappa = (k_1, \ldots, k_r) \in \intpart(n,r)$ and each $i \in \{1, \ldots, r\}$, the integral over $\ket{\alpha_i}_{A}$ of $\ketbra{\alpha_i}_{A}^{\otimes k_i}$ is a direct application of \cref{eq:integral_sym_proj}, yielding the term $\frac{\Pi_{\Sym^{k_i}(A)}}{\dsym_{A}(k_i)}$.
    The same argument applies to the integrals over $\ket{\beta_i}_{B}$, giving $\frac{ \Pi_{\Sym^{k_i}(B)}}{\dsym_{B}(k_i)}$. 
    It remains to evaluate the integral over $\ket{\gamma}_{R}$ of the multinomial distribution $M_p(k_1,\ldots, k_r)$, with probability distribution $p_i = \abs{\braket{i|\gamma}_{R}}^2$. For this, consider the state
    \begin{equation}
        \ket{\Psi_k}_{R^{n}} \coloneqq \ket{1}_{R}^{\otimes k_1} \ot \cdots \ot \ket{r}_{R}^{\otimes k_r} \in R^{\ot n},
    \end{equation}
    so that 
    \begin{equation}
        p_1^{k_1}\cdots p_r^{k_r} = \abs{\braket{\Psi_k|\gamma}^{\ot n}_{R}}^2.
    \end{equation}
    Although $\ket{\Psi_k}$ is not, in general, $S_n$-symmetric, it is invariant under $S_{k_1} \times \cdots \times S_{k_r}$.
    Further, as $\{\ket{1}_{R}, \ldots, \ket{r}_{R}\} \subseteq R$ was assumed orthonormal, we obtain
    \begin{equation}
        \bra{\Psi_k}\Pi_{\Sym^{n}(R)}\ket{\Psi_k}_{R^{n}} = \frac{k_1!\cdots k_r!}{n!}.
    \end{equation}
    Using this, we compute the integral over $\ket \gamma_{R}$ and obtain
    \begin{align}
    \begin{split}
        \label{eq:dirichlet_moment}
        \int d \gamma M_{p}(k_1,\ldots, k_r)
        &= \frac{n!}{k_1!\cdots k_r!}\int d \gamma p_1^{k_1}\cdots p_r^{k_r} \\
        &= \frac{n!}{k_1!\cdots k_r!}\int d \gamma \bra{\Psi_k} (\ketbra{\gamma}_{R})^{\ot n} \ket{\Psi_k} \\
        &= \frac{n!}{k_1!\cdots k_r!}\bra{\Psi_k} \frac{\Pi_{(n)}^{r}}{\dsym_R(n)} \ket{\Psi_k} \\
        &= \frac{1}{\dsym_R(n)},
    \end{split}
    \end{align}
    which completes the proof.
\end{proof}

\subsection{Separable state inequalities}
\label{sec:necessary_inequalities}

We now prove \cref{thm:forward}, which establishes a universal operator upper bound on $\rho_{AB}^{\otimes n}$ for all separable states $\rho_{AB} \in \SEPAB$ and all positive integers $n$.

To illustrate our technique, we begin by describing how the method works for the case of \textit{pure} product states $\rho_{AB} = \ketbra{\Psi}_{AB}$ where $\ket{\Psi}_{AB} = \ket{\alpha}_{A} \otimes \ket{\beta}_{B}$.
The core intuition of our approach is to notice that pure product states $\ket{\Psi}_{AB} = \ket{\psi}_{A} \ot \ket{\phi}_{B}$ have more \textit{symmetry} than pure entangled states;
specifically, while the $n$-th tensor power of a pure state, $\ket{\Psi}_{AB}^{\ot n}$, lies in the symmetric subspace $\Sym^{n}(A \ot B)$, the $n$-th tensor power of a pure \textit{product} state, $\ket{\Psi}_{AB}^{\ot n} = \ket{\alpha}_{A}^{\ot n} \ot \ket{\beta}_{B}^{\ot n}$, lies in the smaller \textit{locally-symmetric subspace} $\Sym^{n}(A) \ot \Sym^{n}(B)$.

This extra symmetry of pure product states can be encoded into an inequality as follows. 

First, let $\Pi_{\Sym^{n}(H)}$ denote the projection operator onto the symmetric subspace $\Sym^{n}(H)$, so that $\Pi_{\Sym^{n}(A)} \ot \Pi_{\Sym^{n}(B)}$ denotes the projection operator onto the corresponding locally-symmetric subspace, $\Sym^{n}(A) \ot \Sym^{n}(B)$.
Then, the vector $\ket{\Psi}_{AB}^{\ot n}$ lies in the locally-symmetric subspace if and only if for every $n \in \mathbb{N}$,
\begin{equation}
    \label{eq:pure_main_ineq}
    \ketbra{\Psi}_{AB}^{\ot n} \preceq \Pi_{\Sym^n(A)} \ot \Pi_{\Sym^n(B)}.
\end{equation}
Note that for this case, it actually suffices to check $n=2$.

To extend this intuition from pure states $\ket{\Psi}_{AB}$ to mixed states $\rho_{AB}$, we use the multinomial expansion given by \cref{prop:multinomial_expansion} and obtain the following result.
\begin{thm}
    \label{thm:stronger_forward}
    If $\rho_{AB} \in\SEPAB$, then for all $n \in \mathbb{N}$,
    \begin{equation}
        \label{eq:stronger_forward}
        \rho_{AB}^{\otimes n} \preceq \Lambda_{A^{n}B^{n}},
    \end{equation}
    where $\Lambda_{A^{n}B^{n}}$ is given by
    \begin{equation}
        \label{eq:lambda_upper_bound}
        \Lambda_{A^{n}B^{n}} = \sum_{\kappa \in \intpart(n,r)} \twirl_{S_n}\left(\Pi_{\Sym^{\kappa}(A)} \otimes \Pi_{\Sym^{\kappa}(B)}\right),
    \end{equation}
    with $r = d_{AB}^{2}$.
\end{thm}
\begin{proof}
    Let $\rho_{AB} \in \SEPAB$ be a separable state with a decomposition as in \cref{eq:sep_r_ab}, namely
    \begin{equation}
        \rho_{AB} = \sum_{i=1}^{r} p_i \ketbra{\alpha_i}_{A} \otimes \ketbra{\beta_i}_{B}.
    \end{equation}
    By \cref{eq:sep_rank_finite}, we may take $r = d_{AB}^{2}$ without loss of generality.
    By the multinomial expansion formula from \cref{prop:multinomial_expansion}, the $n$-th tensor power of $\rho_{AB}$ can be written as
 \begin{equation}
        \rho_{AB}^{\ot n}
        = \sum_{\kappa \in \intpart(n,r)} M_{p}(\kappa) \twirl_{S_n}
        \left(\bigotimes_{i=1}^{r}\ketbra{\alpha_i}_A^{\otimes k_i} \ot \ketbra{\beta_i}_B^{\otimes k_i} \right), \\
    \end{equation}
    where $M_p(\kappa)$ denotes the multinomial probability associated with the probability vector $p = (p_1,\dots,p_r)$. We now use the operator bound
    \begin{equation}
        \bigotimes_{i=1}^{r}\ketbra{\alpha_i}_A^{\otimes k_i} \preceq \bigotimes_{i=1}^{r}\Pi_{\Sym^{k_i}(A)} = \Pi_{\Sym^{\kappa}(A)}.
    \end{equation}
    An analogously bound holds for the for the $\ket{\beta_i}_B$ factors.
    Therefore, we obtain 
    \begin{equation}
        \rho_{AB}^{\ot n} \preceq \sum_{\kappa \in \intpart(n,r)} M_{p}(\kappa) \twirl_{S_n}
        \left(\Pi_{\Sym^{\kappa}(A)} \ot \Pi_{\Sym^{\kappa}(B)} \right).
    \end{equation}
    Finally, since the multinomial distribution probability satisfies
    \begin{equation}
        \label{eq:multinomial_bound}
        M_p(\kappa) \leq 1,
    \end{equation}
    for every $\kappa \in \intpart(n,r)$, and $r = d_{AB}^2$, we obtain the operator upper bound $\Lambda_{A^{n}B^{n}}$ as claimed.
\end{proof}
Before proceeding, we remark the bound in \cref{thm:stronger_forward} could be improved further.
For example, in \cref{eq:multinomial_bound} we bounded the multinomial probability $M_p(\kappa)$ by one. 
Alternatively, one could also bound $M_{p}(k_1,\ldots,k_r)$ by the maximum likelihood probability, $M_{q^{*}}(k_1,\ldots,k_r)$ where $q^{*} = (\frac{k_1}{n},\ldots, \frac{k_r}{n})$, or even group the weak compositions $\kappa$ related by permutation to obtain a tighter version of \cref{thm:stronger_forward}.

Nevertheless, the operator $\Lambda_{A^{n}B^{n}}$ admits a straightforward comparison with the de Finetti operator $\Omega_{A^{n}B^{n}}$ from \cref{thm:forward} in \cref{sec:random_sep_states}, which is sufficient for our purposes.
By direct comparison of their formulas, we see that $\Lambda_{A^{n}B^{n}}$ and $\Omega_{A^{n}B^{n}}$ differ only by normalization factors involving dimensions of the symmetric subspaces.
\begin{lem}
    \label{lem:omega_is_poly_lambda}
    For all $n$, the operators $\Omega_{A^{n}B^{n}}$ and $\Lambda_{A^{n}B^{n}}$ satisfy
    \begin{equation}
        \Omega_{A^{n}B^{n}} \preceq \Lambda_{A^{n}B^{n}} \preceq f_{AB}(n) \Omega_{A^{n}B^{n}},
    \end{equation}
    where 
    \begin{equation}
        \label{eq:f_definition}
        f_{AB}(n) = \dsym_{R}(n)\dsym_{A}\left(\left\lceil\frac{n}{d_{R}}\right\rceil\right)^{d_{R}}\dsym_{B}\left(\left\lceil\frac{n}{d_{R}}\right\rceil\right)^{d_{R}},
    \end{equation}
    with $d_{R} = d_{AB}^2$ and $\dsym_{H}(n) = \dim \Sym^{n}(H)$.
\end{lem}
\begin{proof}
    These inequalities follow by noting that the definition for $\Lambda_{A^{n}B^{n}}$ in \cref{eq:lambda_upper_bound} and the formula for $\Omega_{A^{n}B^{n}}$ in \cref{eq:omega_formula}
    differ only by factors involving dimensions of symmetric subspaces.
    The first inequality holds because each of these dimension factors are at least one, i.e., $\dsym_R(n), \dsym_{A}(\kappa), \dsym_{B}(\kappa) \geq 1$. 
    
    For the second inequality, \cref{lem:blockwise_bound} implies, for any $\kappa = (k_1,\ldots, k_r) \in \intpart(n,r)$, that
    \begin{equation}
        \dsym_{H}(\kappa) \leq \dsym_{H}\left(\left\lceil\frac{n}{r}\right\rceil\right)^r. 
    \end{equation}
    Therefore, by letting $f_{AB}(n)$ be as in \cref{eq:f_definition} with $r = d_{R} = d_{AB}^2$, we see
    \begin{align}
    \begin{split}
        f_{AB}(n) \Omega_{A^{n}B^{n}}
        &= \underbrace{\frac{\dsym_{R}(n)}{ \dsym_{R}(n)}}_{=1} \sum_{\kappa \in \intpart(n,r)} 
        \underbrace{\frac{\dsym_{A}\left(\left\lceil\frac{n}{r}\right\rceil\right)^r}{\dsym_{A}(\kappa)}}_{\geq 1}
        \underbrace{\frac{\dsym_{B}\left(\left\lceil\frac{n}{r}\right\rceil\right)^r}{\dsym_{B}(\kappa)}}_{\geq 1}
        \twirl_{S_n}\left(\Pi_{\Sym^{\kappa}(A)} \ot \Pi_{\Sym^{\kappa}(B)}\right) \\  
        &\succeq \sum_{\kappa \in \intpart(n,r)} \twirl_{S_n}\left(\Pi_{\Sym^{\kappa}(A)} \ot \Pi_{\Sym^{\kappa}(B)}\right)
        = \Lambda_{A^{n}B^{n}}.
    \end{split}
    \end{align}
    This proves the second inequality.
\end{proof}
To prove completeness of our criteria in subsequent sections, we will additionally exploit that the prefactor $f_{AB}(n)$ in \cref{lem:omega_is_poly_lambda} grows at most polynomially in $n$.
To see this clearly, it suffices to consider the following crude bound:
\begin{align}
    \label{eq:f_n_poly_bound}
    \begin{split}
    f_{AB}(n) 
    &= \dsym_{R}(n)\dsym_{A}\left(\left\lceil\frac{n}{d_{R}}\right\rceil\right)^{d_{R}}\dsym_{B}\left(\left\lceil\frac{n}{d_{R}}\right\rceil\right)^{d_{R}}, \\
    &\leq \dsym_{R}(n)\dsym_{A}\left(n\right)^{d_{R}}\dsym_{B}\left(n\right)^{d_{R}}, \\
    &\leq \dsym_{R}(n)\dsym_{AB}\left(n\right)^{d_{R}}, \\
    &\leq (n+1)^{d_{R}-1}(n+1)^{d_{R}(d_{AB}-1)}, \\
    &= (n+1)^{d_{AB}^{3}-1},
    \end{split}
\end{align}
where we again use that $d_{R} = d_{AB}^2$ together with the simple bound $\dsym_H(n) \leq (n+1)^{d_{H}-1}$ \cite{harrow2013church,christandl2009postselection}.
Of course, in practice, the prefactor $f_{AB}(n)$ is typically smaller than the bound provided by \cref{eq:f_n_poly_bound}.

With \cref{lem:omega_is_poly_lambda} established, we now observe that \cref{thm:stronger_forward} implies \cref{thm:forward}, which we restate and prove now.
\forward*
\begin{proof}
    If $\rho_{AB}$ is separable, then \cref{thm:stronger_forward} implies $\rho_{AB}^{\otimes n} \preceq \Lambda_{A^nB^n}$ and therefore \cref{lem:omega_is_poly_lambda} implies,
    \begin{equation}
        \rho_{AB}^{\otimes n} \preceq \Lambda_{A^nB^n}\preceq f_{AB}(n)\Omega_{A^nB^n}.
    \end{equation}
    This observation, combined with the bound on $f_{AB}(n)$ in \cref{eq:f_n_poly_bound} is the statement of \cref{thm:forward}.
\end{proof}

\subsection{Bounding fidelity of separability}

In this section, we establish \cref{thm:completeness} which provides a bound on the logarithmic fidelity of separability, $G(\rho_{AB}) = - \log F_s(\rho_{AB})$, whenever $\rho_{AB}$ satisfies the operator inequality $\rho_{AB}^{\otimes n} \leq f_{AB}(n) \Omega_{A^{n}B^{n}}$ in \cref{thm:forward}.
The proof relies on \cref{lem:likelihood_ratio_fidelity}, which compares the fidelity between any two states $\rho$ and $\sigma$ using a sequence of measurement effects $M_n$ depending only on $\rho$.
For convenience, we restate \cref{thm:completeness} below.
\completeness*
\begin{proof}
    For any state $\rho_{AB} \in \mathcal S(A\otimes B)$, by \cref{lem:likelihood_ratio_fidelity} there exists a sequence of measurements $M_{A^nB^n} \succeq 0$ such that, for all $n \in \mathbb N$ and all $\sigma_{AB} \in \mathcal S(A\otimes B)$,
    \begin{equation}
        \mathcal R_{\sigma_{AB}/\rho_{AB}}(M_{A^nB^n})= \frac{\Tr(M_{A^nB^n} \sigma_{AB}^{\otimes n})}{\Tr(M_{A^nB^n} \rho_{AB}^{\otimes n})} \leq C(\rho_{AB})^{-1} F(\rho_{AB}, \sigma_{AB})^{n - d_{AB}^2},
    \end{equation}
    where $C(\rho_{AB})$ is the constant defined in \cref{eq:staircase_constant}.
    Maximizing over all separable states $\sigma_{AB} \in \SEPAB$ gives
    \begin{equation}
        \sup_{\sigma_{AB} \in \SEP}\mathcal R_{\sigma_{AB}/\rho_{AB}}(M_{A^nB^n}) \leq C(\rho_{AB})^{-1} F_{s}(\rho_{AB})^{n - d_{AB}^2},
    \end{equation}
    where $F_{s}(\rho_{AB}) = \sup_{\sigma_{AB} \in \SEP} F(\rho_{AB}, \sigma_{AB})$ is the fidelity of separability of $\rho_{AB}$.
    Now assume that $\rho_{AB}$ satisfies the operator inequality $\rho_{AB}^{\otimes n} \leq f_{AB}(n)\Omega_{A^nB^n}$ for some $n > d_{AB}^2$. 
    Then, applying the measurement effect $M_{A^{n}B^{n}}$ yields
    \begin{equation}
        \Tr(M_{A^nB^n}\rho_{AB}^{\otimes n}) \leq  f_{AB}(n) \Tr(M_{A^nB^n}  \Omega_{A^{n}B^{n}}) \leq  f_{AB}(n)\sup_{\sigma_{AB}\in \SEP}   \Tr(M_{A^nB^n} \sigma_{AB}^{\otimes n}).
    \end{equation}
    This implies the lower bound
    \begin{equation}
        \sup_{\sigma_{AB} \in \SEP} \mathcal R_{\sigma_{AB}/\rho_{AB}}(M_{A^nB^n}) \geq \frac{1}{ f_{AB}(n)}.
    \end{equation}
    Combining these two bounds on the maximum likelihood ratio, and noting $n-d_{AB}^{2} > 0$, yields
    \begin{equation}
        F_{s}(\rho_{AB}) \geq \left(\frac{1}{C(\rho_{AB})^{-1}  f_{AB}(n)}\right)^{\frac{1}{n-d_{AB}^{2}}},
    \end{equation}
    and therefore
    \begin{equation}
        G(\rho_{AB}) = - \log F_{s}(\rho_{AB}) \leq \frac{\log (C(\rho_{AB})^{-1}  f_{AB}(n))}{n-d_{AB}^2}.
    \end{equation}
\end{proof}
From \cref{thm:completeness}, we conclude that the criteria in \cref{thm:forward} is complete: any state satisfying the tensor-power bound for arbitrarily large $n$ must have vanishing logarithmic fidelity of separability, and hence must be separable.
\begin{cor}
    \label{cor:reverse}
    If $\rho_{AB} \in \mathcal S(A \otimes B)$ satisfies $\rho_{AB}^{\otimes n} \preceq f_{AB}(n)\Omega_{A^{n} B^{n}}$ for every $n \in \mathbb N$, then $\rho_{AB} \in \SEPAB$.
\end{cor}
\begin{proof}
    By \cref{thm:completeness}, we have $G(\rho_{AB}) \leq O(\log n / n)$, since the pre-factor $f_{AB}(n)$ grows polynomially in $n$ by \cref{eq:f_n_poly_bound}.
    Taking the limit $n \to \infty$, we obtain $G(\rho_{AB}) = 0$. Since $G(\rho_{AB})$ is faithful, this implies that $\rho_{AB} \in \SEPAB$.
\end{proof}

\subsection{Entanglement witnesses}
\label{sec:poly_ent_witnesses}
The main application of our results is the construction of a complete family of polynomial entanglement witnesses that are, in addition, invariant under local unitaries. This establishes \cref{cor:complete_witnesses}.

We begin by observing that the entanglement criteria provided by \cref{thm:stronger_forward} are operator inequalities of the form $\rho_{AB}^{\otimes n} \preceq \Lambda_{A^nB^n}$, for some operator $\Lambda_{A^nB^n}$. As such, they are \textit{not} polynomial entanglement witnesses, which instead correspond to polynomial inequalities of the form
\begin{equation}
    p(\rho_{AB}) = \Tr(W_{A^{\ell}B^{\ell}} \rho_{AB}^{\otimes \ell}) \geq 0, 
\end{equation}
where $W_{A^{\ell}B^{\ell}}$ is a Hermitian operator on $(A\otimes B)^{\otimes \ell}$ for some $\ell \in \mathbb N$.
However, it is possible to extract polynomial entanglement witnesses from operator inequalities in several ways.
The simplest construction is to take any positive semidefinite operator $M_{A^nB^n} \succeq 0$ on $(A\otimes B)^{\otimes n}$, and then to let
\begin{equation}
    W_{A^{n}B^{n}} = \Tr(M_{A^nB^n}\Lambda_{A^{n}B^{n}})I_{AB}^{\ot n} - M_{A^nB^n}.
\end{equation}
The operator inequality $\rho_{AB}^{\otimes n} \preceq \Lambda_{A^nB^n}$ for $\rho_{AB} \in \SEPAB$ then implies that $\Tr(W_{A^{n}B^{n}} \rho_{AB}^{\ot n})$ is non-negative, because
\begin{equation}
    \Tr(W_{A^{n}B^{n}} \rho_{AB}^{\ot n}) = \Tr(M_{A^nB^n}(\Lambda_{A^{n}B^{n}}-\rho_{AB}^{\ot n})) \geq 0.
\end{equation}
One consequence of \cref{thm:stronger_forward} and \cref{thm:completeness} is that \textit{every} entangled state $\rho_{AB} \not \in \SEPAB$ can be detected by a witness of this form, for some sufficiently large $n$ and some operator $M_{A^nB^n} \succeq 0$.
Indeed, it suffices to choose $M_{A^nB^n}$ to be any measurement effect capable of tomographically distinguishing $\rho_{AB}$ from other states $\sigma_{AB}$, for instance, the $\rho_{AB}$-dependent measurement effects provided by \cref{lem:likelihood_ratio_fidelity}.
For the purposes of \textit{experimental} entanglement detection, however, this approach has a drawback: when $\rho_{AB}$ is unknown, it is unclear which measurement operator $M_{A^nB^n}$ and thus which witnesses $W_{A^{n}B^{n}}$ suffices to certify the entanglement of $\rho_{AB}$.

To address this issue, we adopt a different approach based on Sylvester's criterion for operator positivity, which leads to a family of entanglement witnesses that is both independent of the state $\rho_{AB}$ and also complete.
This criterion states that an $\ell \times \ell$ Hermitian matrix $X$ is positive semidefinite, $X \succeq 0$, if and only if, for every $m \in \{1,\ldots, \ell\}$ the polynomial
\begin{equation}
    E_m(X) = \Tr(P_{\wedge^{m}\mathbb C^{\ell}} X^{\otimes m}) 
\end{equation}
is nonnegative, $E_m(X)\geq 0$. Here, $P_{\wedge^{m}\mathbb C^{\ell}}$
denotes the projection onto the \textit{antisymmetric} subspace $\wedge^{m}\mathbb C^{\ell}$ of $(\mathbb C^{\ell})^{\otimes m}$ (see \cref{prop:sylvesters} in \cref{sec:sylvester} for a derivation).

The key to converting the operator inequality in \cref{thm:stronger_forward} into polynomial entanglement witnesses is to evaluate $E_m$ on $X = \Lambda_{A^{n}B^{n}} - \rho_{AB}^{\otimes n}$. Thus, for each $n \in \mathbb N$ and $m \in \{1, \ldots, d_{AB}^{n}\}$, we define the degree-$nm$ polynomial 
\begin{equation}
    p_{n,m}(\rho_{AB}) \coloneqq E_m(\Lambda_{A^{n}B^{n}} - \rho_{AB}^{\ot n}).
\end{equation}
Since $E_m(\Lambda_{A^{n}B^{n}} - \rho_{AB}^{\ot n})$ is a degree-$nm$ polynomial in the coefficients of $\rho_{AB}$, there necessarily exists a Hermitian operator $W_{A^{nm}B^{nm}}^{n,m}$ on $(A\otimes B)^{\otimes nm}$, independent of $\rho_{AB}$, such that for all $\rho_{AB}$ we have,
\begin{equation}
    \label{eq:sylvester_to_witness}
    p_{n,m}(\rho_{AB}) = \Tr(W_{A^{nm}B^{nm}}^{n,m} \rho_{AB}^{\otimes nm}).
\end{equation}
For an explicit formula for $W_{A^{nm}B^{nm}}^{n,m}$ as a function of $\Lambda_{A^{n}B^{n}}$ and $m$, we refer the reader to 
\cref{lem:witness_from_lambda} in \cref{sec:sylvester}.
This method of converting operator inequalities into polynomial inequalities is inspired by its application to $X = \rho_{AB}^{T_{B}}$, where it is used to characterize PPT-ness via polynomial entanglement witnesses \cite{elben2020mixed,neven2021symmetry}.

In any case, we obtain another characterization of $\SEPAB$, now in terms of a countable family of polynomial entanglement witnesses.
\begin{cor}
    \label{cor:complete_entanglement_witnesses}
    A state $\rho_{AB}$ belongs to $\SEPAB$ if and only if, for all $n,m \in \mathbb N$ with $m \leq d_{AB}^{n}$,
    \begin{equation}
        p_{n,m}(\rho_{AB}) \geq 0,
    \end{equation}
    where $p_{n,m}(\rho_{AB}) = E_m(\Lambda_{A^{n}B^{n}} - \rho_{AB}^{\ot n})$.
\end{cor}
Moreover, the polynomial entanglement witnesses in \cref{cor:complete_entanglement_witnesses} are local-unitary invariant. Indeed, this follows from the fact that the operator $\Lambda_{A^{n}B^{n}}$ from \cref{thm:stronger_forward} is invariant under local unitaries, and the polynomial $E_m$ is $U(d_{AB})$-invariant.

\section{Discussion}
\label{sec:discussion}

In this work, we provided a systematic characterization of the set of separable states in arbitrary local dimensions $(d_{A}, d_{B})$, both in terms of operator inequalities (\cref{thm:stronger_forward}) and in terms of entanglement witnesses (\cref{cor:complete_entanglement_witnesses}). We now address some of the advantages and limitations of our results, and outline directions for future investigation.

First and foremost, \cref{cor:complete_entanglement_witnesses} involves a countably-infinite family of polynomials of arbitrarily high degree in $\rho_{AB}$.
As there are intrinsic difficulties in evaluating high-degree polynomials or in measuring many copies of a quantum state, evaluating the entanglement witnesses in \cref{cor:complete_entanglement_witnesses}, or the operator inequalities in \cref{thm:stronger_forward}, becomes infeasible for larger and larger values of $n$. 
Indeed, from the definition of $\Lambda_{A^{n}B^{n}}$ in \cref{eq:lambda_upper_bound} it can be seen that the operator inequality in \cref{thm:stronger_forward} becomes a \textit{trivial} inequality when $n \leq d_{AB}^2$, which implies that the number of copies $n$ must always be taken larger than the separable rank $r$ which, in turn, can be as large as $d_{AB}^2$.\footnote{If $n, r \in \mathbb N$ are such that $n \leq r$ then $\intpart(n,r)$ contains the weak composition $\kappa = (\underbrace{1,\ldots,1}_{n},\underbrace{0,\ldots,0}_{r-n})$ for which $\Sym^{\kappa}(H) \equiv H^{\otimes n}$ and then $\Lambda_{A^{n}B^{n}} \succeq I_{AB}^{\ot n} \succeq \rho_{AB}^{\otimes n}$ holds for \textit{all} bipartite states $\rho_{AB}$.}
Consequently, we anticipate that families of low-degree entanglement witnesses, such as in \cite{miller2026detecting}, will remain more efficient and practically relevant for experimental entanglement detection, despite their incompleteness in higher dimensions.

Fortunately, for fixed dimensions $(d_{A}, d_{B})$ the set of separable states is \textit{semialgebraic}, meaning the set $\SEPAB$ can, in theory, be expressed using a \textit{finite} number of polynomial inequalities in the coefficients of $\rho_{AB}$.
In other words, for fixed local dimensions $(d_{A}, d_{B})$, theoretically there exists a \textit{finite} collection of polynomial entanglement witnesses which is complete for detecting entanglement.
In \cref{sec:semialgebraic}, we additionally prove that this semialgebraic property of $\SEPAB$ extends to the orbit space $\SEPAB/(U(d_{A}) \times U(d_{B}))$, meaning that there also exists a finite and complete family of entanglement witnesses consisting entirely of local-unitary invariant
polynomials.
Unfortunately, while algorithms for finding such polynomials exist \cite{jirstrand1995cylindrical}, they are computationally prohibitive -- so finding a complete family of entanglement witnesses that is \textit{finite}, even for $(d_{A}, d_{B}) = (3,3)$ or $(2,4)$, remains a challenging open problem.

Third, while we have only stated our results in terms of bipartite separability, our techniques can be readily extended to address the case of \textit{fully-separable} multipartite states.
For example, using the techniques in this paper one can see that a tripartite state $\rho_{ABC}$ is fully-separable,
\begin{equation}
    \rho_{ABC} = \sum_{i=1}^{r} p_i \ketbra{\alpha_i}_{A} \ot \ketbra{\beta_i}_{B} \otimes \ketbra{\gamma_i}_{C},
\end{equation}
with $r = d_{ABC}^2$, if and only if, for all $n \in \mathbb N$, it satisfies $\rho_{ABC}^{\ot n} \preceq \Lambda_{A^nB^nC^n}$, where
\begin{equation}
    \Lambda_{A^nB^nC^n} = \sum_{\kappa \in \intpart(n,r)} \twirl_{S_n} \left(\Pi_{\Sym^{\kappa}(A)} \ot \Pi_{\Sym^{\kappa}(B)} \ot \Pi_{\Sym^{\kappa}(C)}\right).
\end{equation}
While our methods can also be adapted to handle some additionally more nuanced notions of multipartite entanglement such as genuine multipartite entanglement \cite{dur1999separability,dur2000three,guhne2010separability}, it is presently unclear how to extend the techniques of this paper to derive an analogous complete characterization for the set of bi-separable states~\cite{guhne2009entanglement}, or $k$-partite states which are accessible via local operations, global shared randomness, and $(k-1)$-shared entanglement (LOSRSE) as proposed in \cite{schmid2023understanding}\footnote{Although, a promising approach would be to combine the techniques developed here with both the fully quantum inflation technique \cite{smith2026fully} and the quantum marginal inequalities from \cite{fraser2022sufficient}.}.

Finally, one may be concerned that the inequalities in \cref{thm:stronger_forward} and \cref{cor:complete_entanglement_witnesses} involve $(d_{AB})^{n}$-dimensional matrices, which grow rapidly with increasing $n$.
To address this issue, in \cref{sec:fourier_transform}, we show how to block-diagonalize \cref{thm:stronger_forward} into a polynomial number of blocks, of polynomial size, using the $S_n$-invariance of \cref{thm:stronger_forward}.
Moreover, \cref{sec:fourier_transform} provides a manifestly local-unitary invariant formulation of \cref{thm:stronger_forward}, in the form of \cref{cor:fourier_transform_ineq}.

\subsection*{Acknowledgments}
TCF and VB acknowledge financial support from VILLUM FONDEN via the QMATH Centre of Excellence (Grant No. 10059) and the Novo Nordisk Foundation (grant NNF20OC0059939 “Quantum for Life”).
RR acknowledges financial support from the ERC grant GIFNEQ 101163938.
FH was funded in whole or in part by the National Science Centre, Poland 2024/54/E/ST2/00451
and by the Polish National Agency for Academic Exchange under the Strategic Partnership Programme grant BNI/PST/2023/1/00013/U/00001.
For the purpose of Open Access,
the authors have applied a CC-BY public copyright licence to any Author Accepted Manuscript (AAM) version arising from this submission.
The authors thank Alexander Blomenhofer, Matthias Christandl, Alexander Mueller-Hermes, Harold Nieuwboer, Elias Theil and Michael M. Wolf for helpful discussions.

\bibliography{references}


\begin{thebibliography}{52}
\ifx \bisbn   \undefined \def \bisbn  #1{ISBN #1}\fi
\ifx \binits  \undefined \def \binits#1{#1}\fi
\ifx \bauthor  \undefined \def \bauthor#1{#1}\fi
\ifx \batitle  \undefined \def \batitle#1{#1}\fi
\ifx \bjtitle  \undefined \def \bjtitle#1{#1}\fi
\ifx \bvolume  \undefined \def \bvolume#1{\textbf{#1}}\fi
\ifx \byear  \undefined \def \byear#1{#1}\fi
\ifx \bissue  \undefined \def \bissue#1{#1}\fi
\ifx \bfpage  \undefined \def \bfpage#1{#1}\fi
\ifx \blpage  \undefined \def \blpage #1{#1}\fi
\ifx \burl  \undefined \def \burl#1{\textsf{#1}}\fi
\ifx \doiurl  \undefined \def \doiurl#1{\url{https://doi.org/#1}}\fi
\ifx \betal  \undefined \def \betal{\textit{et al.}}\fi
\ifx \binstitute  \undefined \def \binstitute#1{#1}\fi
\ifx \binstitutionaled  \undefined \def \binstitutionaled#1{#1}\fi
\ifx \bctitle  \undefined \def \bctitle#1{#1}\fi
\ifx \beditor  \undefined \def \beditor#1{#1}\fi
\ifx \bpublisher  \undefined \def \bpublisher#1{#1}\fi
\ifx \bbtitle  \undefined \def \bbtitle#1{#1}\fi
\ifx \bedition  \undefined \def \bedition#1{#1}\fi
\ifx \bseriesno  \undefined \def \bseriesno#1{#1}\fi
\ifx \blocation  \undefined \def \blocation#1{#1}\fi
\ifx \bsertitle  \undefined \def \bsertitle#1{#1}\fi
\ifx \bsnm \undefined \def \bsnm#1{#1}\fi
\ifx \bsuffix \undefined \def \bsuffix#1{#1}\fi
\ifx \bparticle \undefined \def \bparticle#1{#1}\fi
\ifx \barticle \undefined \def \barticle#1{#1}\fi
\bibcommenthead
\ifx \bconfdate \undefined \def \bconfdate #1{#1}\fi
\ifx \botherref \undefined \def \botherref #1{#1}\fi
\ifx \url \undefined \def \url#1{\textsf{#1}}\fi
\ifx \bchapter \undefined \def \bchapter#1{#1}\fi
\ifx \bbook \undefined \def \bbook#1{#1}\fi
\ifx \bcomment \undefined \def \bcomment#1{#1}\fi
\ifx \oauthor \undefined \def \oauthor#1{#1}\fi
\ifx \citeauthoryear \undefined \def \citeauthoryear#1{#1}\fi
\ifx \endbibitem  \undefined \def \endbibitem {}\fi
\ifx \bconflocation  \undefined \def \bconflocation#1{#1}\fi
\ifx \arxivurl  \undefined \def \arxivurl#1{\textsf{#1}}\fi
\csname PreBibitemsHook\endcsname

\bibitem[\protect\citeauthoryear{Gurvits}{2004}]{gurvits2004classical}
\begin{barticle}
\bauthor{\bsnm{Gurvits}, \binits{L.}}:
\batitle{Classical complexity and quantum entanglement}.
\bjtitle{Journal of Computer and System Sciences}
\bvolume{69}(\bissue{3}),
\bfpage{448}--\blpage{484}
(\byear{2004})
\doiurl{10.1016/j.jcss.2004.06.003}
\end{barticle}
\endbibitem

\bibitem[\protect\citeauthoryear{Aubrun and Szarek}{2015}]{aubrun2015dvoretzky}
\begin{barticle}
\bauthor{\bsnm{Aubrun}, \binits{G.}},
\bauthor{\bsnm{Szarek}, \binits{S.}}:
\batitle{Dvoretzky's theorem and the complexity of entanglement detection}.
\bjtitle{arXiv:1510.00578}
(\byear{2015})
\doiurl{10.48550/arXiv.1510.00578}
\end{barticle}
\endbibitem

\bibitem[\protect\citeauthoryear{Peres}{1996}]{peres1996separability}
\begin{barticle}
\bauthor{\bsnm{Peres}, \binits{A.}}:
\batitle{Separability criterion for density matrices}.
\bjtitle{Physical Review Letters}
\bvolume{77}(\bissue{8}),
\bfpage{1413}
(\byear{1996})
\doiurl{10.1103/PhysRevLett.77.1413}
\end{barticle}
\endbibitem

\bibitem[\protect\citeauthoryear{Doherty et~al.}{2004}]{doherty2004complete}
\begin{barticle}
\bauthor{\bsnm{Doherty}, \binits{A.C.}},
\bauthor{\bsnm{Parrilo}, \binits{P.A.}},
\bauthor{\bsnm{Spedalieri}, \binits{F.M.}}:
\batitle{Complete family of separability criteria}.
\bjtitle{Physical Review A}
\bvolume{69}(\bissue{2}),
\bfpage{022308}
(\byear{2004})
\doiurl{10.1103/physreva.69.022308}
\end{barticle}
\endbibitem

\bibitem[\protect\citeauthoryear{Navascu{\'e}s et~al.}{2009}]{navascues2009complete}
\begin{barticle}
\bauthor{\bsnm{Navascu{\'e}s}, \binits{M.}},
\bauthor{\bsnm{Owari}, \binits{M.}},
\bauthor{\bsnm{Plenio}, \binits{M.B.}}:
\batitle{Complete criterion for separability detection}.
\bjtitle{Physical review letters}
\bvolume{103}(\bissue{16}),
\bfpage{160404}
(\byear{2009})
\doiurl{10.1103/PhysRevLett.103.160404}
\end{barticle}
\endbibitem

\bibitem[\protect\citeauthoryear{Harrow et~al.}{2017}]{harrow2017improved}
\begin{barticle}
\bauthor{\bsnm{Harrow}, \binits{A.W.}},
\bauthor{\bsnm{Natarajan}, \binits{A.}},
\bauthor{\bsnm{Wu}, \binits{X.}}:
\batitle{An improved semidefinite programming hierarchy for testing entanglement}.
\bjtitle{Communications in Mathematical Physics}
\bvolume{352}(\bissue{3}),
\bfpage{881}--\blpage{904}
(\byear{2017})
\doiurl{10.1007/s00220-017-2859-0}
\end{barticle}
\endbibitem

\bibitem[\protect\citeauthoryear{Pena et~al.}{2025}]{pena2025tailored}
\begin{barticle}
\bauthor{\bsnm{Pena}, \binits{J.}},
\bauthor{\bsnm{Siddhu}, \binits{V.}},
\bauthor{\bsnm{Tayur}, \binits{S.}}:
\batitle{Tailored first-order and interior-point methods and a new semidefinite programming hierarchy for entanglement detection}.
\bjtitle{arXiv:2508.05854}
(\byear{2025})
\doiurl{10.48550/arXiv.2508.05854}
\end{barticle}
\endbibitem

\bibitem[\protect\citeauthoryear{G{\"u}hne and L{\"u}tkenhaus}{2006}]{guhne2006nonlinear}
\begin{barticle}
\bauthor{\bsnm{G{\"u}hne}, \binits{O.}},
\bauthor{\bsnm{L{\"u}tkenhaus}, \binits{N.}}:
\batitle{Nonlinear entanglement witnesses}.
\bjtitle{Physical Review Letters}
\bvolume{96}(\bissue{17}),
\bfpage{170502}
(\byear{2006})
\doiurl{10.1103/PhysRevLett.96.170502}
\end{barticle}
\endbibitem

\bibitem[\protect\citeauthoryear{Chru{\'s}ci{\'n}ski and Sarbicki}{2014}]{chruscinski2014entanglement}
\begin{barticle}
\bauthor{\bsnm{Chru{\'s}ci{\'n}ski}, \binits{D.}},
\bauthor{\bsnm{Sarbicki}, \binits{G.}}:
\batitle{Entanglement witnesses: construction, analysis and classification}.
\bjtitle{Journal of Physics A: Mathematical and Theoretical}
\bvolume{47}(\bissue{48}),
\bfpage{483001}
(\byear{2014})
\doiurl{10.1088/1751-8113/47/48/483001}
\end{barticle}
\endbibitem

\bibitem[\protect\citeauthoryear{Rico and Huber}{2024}]{rico2024entanglement}
\begin{barticle}
\bauthor{\bsnm{Rico}, \binits{A.}},
\bauthor{\bsnm{Huber}, \binits{F.}}:
\batitle{Entanglement detection with trace polynomials}.
\bjtitle{Physical Review Letters}
\bvolume{132}(\bissue{7}),
\bfpage{070202}
(\byear{2024})
\doiurl{10.1103/PhysRevLett.132.070202}
\end{barticle}
\endbibitem

\bibitem[\protect\citeauthoryear{Elben et~al.}{2020}]{elben2020mixed}
\begin{barticle}
\bauthor{\bsnm{Elben}, \binits{A.}},
\bauthor{\bsnm{Kueng}, \binits{R.}},
\bauthor{\bsnm{Huang}, \binits{H.-Y.}},
\bauthor{\bsnm{Bijnen}, \binits{R.}},
\bauthor{\bsnm{Kokail}, \binits{C.}},
\bauthor{\bsnm{Dalmonte}, \binits{M.}},
\bauthor{\bsnm{Calabrese}, \binits{P.}},
\bauthor{\bsnm{Kraus}, \binits{B.}},
\bauthor{\bsnm{Preskill}, \binits{J.}},
\bauthor{\bsnm{Zoller}, \binits{P.}}, \betal:
\batitle{Mixed-state entanglement from local randomized measurements}.
\bjtitle{Physical Review Letters}
\bvolume{125}(\bissue{20}),
\bfpage{200501}
(\byear{2020})
\doiurl{10.1103/PhysRevLett.125.200501}
\end{barticle}
\endbibitem

\bibitem[\protect\citeauthoryear{Neven et~al.}{2021}]{neven2021symmetry}
\begin{barticle}
\bauthor{\bsnm{Neven}, \binits{A.}},
\bauthor{\bsnm{Carrasco}, \binits{J.}},
\bauthor{\bsnm{Vitale}, \binits{V.}},
\bauthor{\bsnm{Kokail}, \binits{C.}},
\bauthor{\bsnm{Elben}, \binits{A.}},
\bauthor{\bsnm{Dalmonte}, \binits{M.}},
\bauthor{\bsnm{Calabrese}, \binits{P.}},
\bauthor{\bsnm{Zoller}, \binits{P.}},
\bauthor{\bsnm{Vermersch}, \binits{B.}},
\bauthor{\bsnm{Kueng}, \binits{R.}}, \betal:
\batitle{Symmetry-resolved entanglement detection using partial transpose moments}.
\bjtitle{npj Quantum Information}
\bvolume{7}(\bissue{1}),
\bfpage{152}
(\byear{2021})
\doiurl{10.1038/s41534-021-00487-y}
\end{barticle}
\endbibitem

\bibitem[\protect\citeauthoryear{Miller and Eisert}{2026}]{miller2026detecting}
\begin{barticle}
\bauthor{\bsnm{Miller}, \binits{D.}},
\bauthor{\bsnm{Eisert}, \binits{J.}}:
\batitle{Detecting entanglement from few partial transpose moments and their decay via weight enumerators}.
\bjtitle{arXiv:2604.12576}
(\byear{2026})
\doiurl{10.48550/arXiv.2604.12576}
\end{barticle}
\endbibitem

\bibitem[\protect\citeauthoryear{Gulati et~al.}{2024}]{gulati2024using}
\begin{barticle}
\bauthor{\bsnm{Gulati}, \binits{V.}},
\bauthor{\bsnm{Singh}, \binits{G.}},
\bauthor{\bsnm{Dorai}, \binits{K.}}:
\batitle{Using linear and nonlinear entanglement witnesses to generate and detect bound entangled states on an ibm quantum processor}.
\bjtitle{Physica Scripta}
\bvolume{99}(\bissue{11}),
\bfpage{115122}
(\byear{2024})
\doiurl{10.1088/1402-4896/ad87c7}
\end{barticle}
\endbibitem

\bibitem[\protect\citeauthoryear{Grassl et~al.}{1998}]{grassl1998computing}
\begin{barticle}
\bauthor{\bsnm{Grassl}, \binits{M.}},
\bauthor{\bsnm{R{\"o}tteler}, \binits{M.}},
\bauthor{\bsnm{Beth}, \binits{T.}}:
\batitle{Computing local invariants of quantum-bit systems}.
\bjtitle{Physical Review A}
\bvolume{58}(\bissue{3}),
\bfpage{1833}
(\byear{1998})
\doiurl{10.1103/physreva.58.1833}
\end{barticle}
\endbibitem

\bibitem[\protect\citeauthoryear{Leifer et~al.}{2004}]{leifer2004measuring}
\begin{barticle}
\bauthor{\bsnm{Leifer}, \binits{M.S.}},
\bauthor{\bsnm{Linden}, \binits{N.}},
\bauthor{\bsnm{Winter}, \binits{A.}}:
\batitle{Measuring polynomial invariants of multiparty quantum states}.
\bjtitle{Physical Review A}
\bvolume{69}(\bissue{5}),
\bfpage{052304}
(\byear{2004})
\doiurl{10.1103/physreva.69.052304}
\end{barticle}
\endbibitem

\bibitem[\protect\citeauthoryear{Schmid et~al.}{2023}]{schmid2023understanding}
\begin{barticle}
\bauthor{\bsnm{Schmid}, \binits{D.}},
\bauthor{\bsnm{Fraser}, \binits{T.C.}},
\bauthor{\bsnm{Kunjwal}, \binits{R.}},
\bauthor{\bsnm{Sainz}, \binits{A.B.}},
\bauthor{\bsnm{Wolfe}, \binits{E.}},
\bauthor{\bsnm{Spekkens}, \binits{R.W.}}:
\batitle{Understanding the interplay of entanglement and nonlocality: motivating and developing a new branch of entanglement theory}.
\bjtitle{Quantum}
\bvolume{7},
\bfpage{1194}
(\byear{2023})
\doiurl{10.22331/q-2023-12-04-1194}
\end{barticle}
\endbibitem

\bibitem[\protect\citeauthoryear{Horodecki et~al.}{2001}]{horodecki2001separability}
\begin{barticle}
\bauthor{\bsnm{Horodecki}, \binits{M.}},
\bauthor{\bsnm{Horodecki}, \binits{P.}},
\bauthor{\bsnm{Horodecki}, \binits{R.}}:
\batitle{Separability of n-particle mixed states: necessary and sufficient conditions in terms of linear maps}.
\bjtitle{Physics Letters A}
\bvolume{283}(\bissue{1-2}),
\bfpage{1}--\blpage{7}
(\byear{2001})
\doiurl{10.1016/S0375-9601(01)00142-6}
\end{barticle}
\endbibitem

\bibitem[\protect\citeauthoryear{Streltsov et~al.}{2010}]{streltsov2010linking}
\begin{barticle}
\bauthor{\bsnm{Streltsov}, \binits{A.}},
\bauthor{\bsnm{Kampermann}, \binits{H.}},
\bauthor{\bsnm{Bru{\ss}}, \binits{D.}}:
\batitle{Linking a distance measure of entanglement to its convex roof}.
\bjtitle{New Journal of Physics}
\bvolume{12}(\bissue{12}),
\bfpage{123004}
(\byear{2010})
\doiurl{10.1088/1367-2630/12/12/123004}
\end{barticle}
\endbibitem

\bibitem[\protect\citeauthoryear{Gribling et~al.}{2022}]{gribling2022bounding}
\begin{barticle}
\bauthor{\bsnm{Gribling}, \binits{S.}},
\bauthor{\bsnm{Laurent}, \binits{M.}},
\bauthor{\bsnm{Steenkamp}, \binits{A.}}:
\batitle{Bounding the separable rank via polynomial optimization}.
\bjtitle{Linear Algebra and its Applications}
\bvolume{648},
\bfpage{1}--\blpage{55}
(\byear{2022})
\doiurl{10.1016/j.laa.2022.04.010}
\end{barticle}
\endbibitem

\bibitem[\protect\citeauthoryear{Tie et~al.}{2011}]{tie2011rearrangement}
\begin{barticle}
\bauthor{\bsnm{Tie}, \binits{L.}},
\bauthor{\bsnm{Cai}, \binits{K.-Y.}},
\bauthor{\bsnm{Lin}, \binits{Y.}}:
\batitle{Rearrangement inequalities for {H}ermitian matrices}.
\bjtitle{Linear algebra and its applications}
\bvolume{434}(\bissue{2}),
\bfpage{443}--\blpage{456}
(\byear{2011})
\doiurl{10.1016/j.laa.2010.08.043}
\end{barticle}
\endbibitem

\bibitem[\protect\citeauthoryear{Lin and Sra}{2014}]{lin2014completely}
\begin{barticle}
\bauthor{\bsnm{Lin}, \binits{M.}},
\bauthor{\bsnm{Sra}, \binits{S.}}:
\batitle{Completely strong superadditivity of generalized matrix functions}.
\bjtitle{arXiv:1410.1958}
(\byear{2014})
\doiurl{10.48550/arXiv.1410.1958}
\end{barticle}
\endbibitem

\bibitem[\protect\citeauthoryear{Alekseev et~al.}{2025}]{alekseev2025tetrahedral}
\begin{barticle}
\bauthor{\bsnm{Alekseev}, \binits{A.}},
\bauthor{\bsnm{Christandl}, \binits{M.}},
\bauthor{\bsnm{Fraser}, \binits{T.C.}}:
\batitle{The tetrahedral {H}orn problem and asymptotics of {$U(n)$} {$6j$} symbols}.
\bjtitle{arXiv:2510.04877}
(\byear{2025})
\doiurl{10.48550/arXiv.2510.04877}
\end{barticle}
\endbibitem

\bibitem[\protect\citeauthoryear{Harrow}{2013}]{harrow2013church}
\begin{barticle}
\bauthor{\bsnm{Harrow}, \binits{A.W.}}:
\batitle{The church of the symmetric subspace}.
\bjtitle{arXiv:1308.6595}
(\byear{2013})
\doiurl{10.48550/arXiv.1308.6595}
\end{barticle}
\endbibitem

\bibitem[\protect\citeauthoryear{Fraser}{2022}]{fraser2022sufficient}
\begin{barticle}
\bauthor{\bsnm{Fraser}, \binits{T.C.}}:
\batitle{A sufficient family of necessary inequalities for the compatibility of quantum marginals}.
\bjtitle{arXiv:2211.00685}
(\byear{2022})
\doiurl{10.48550/arXiv.2211.00685}
\end{barticle}
\endbibitem

\bibitem[\protect\citeauthoryear{{\.Z}yczkowski and Sommers}{2001}]{zyczkowski2001induced}
\begin{barticle}
\bauthor{\bsnm{{\.Z}yczkowski}, \binits{K.}},
\bauthor{\bsnm{Sommers}, \binits{H.-J.}}:
\batitle{Induced measures in the space of mixed quantum states}.
\bjtitle{Journal of Physics A: Mathematical and General}
\bvolume{34}(\bissue{35}),
\bfpage{7111}
(\byear{2001})
\doiurl{10.1088/0305-4470/34/35/335}
\end{barticle}
\endbibitem

\bibitem[\protect\citeauthoryear{Christandl et~al.}{2009}]{christandl2009postselection}
\begin{barticle}
\bauthor{\bsnm{Christandl}, \binits{M.}},
\bauthor{\bsnm{K{\"o}nig}, \binits{R.}},
\bauthor{\bsnm{Renner}, \binits{R.}}:
\batitle{Postselection technique for quantum channels with applications to quantum cryptography}.
\bjtitle{Physical Review Letters}
\bvolume{102}(\bissue{2}),
\bfpage{020504}
(\byear{2009})
\doiurl{10.1103/physrevlett.102.020504}
\end{barticle}
\endbibitem

\bibitem[\protect\citeauthoryear{Jirstrand}{1995}]{jirstrand1995cylindrical}
\begin{botherref}
\oauthor{\bsnm{Jirstrand}, \binits{M.}}:
Cylindrical algebraic decomposition-an introduction.
Technical report,
Link{\"o}ping University
(1995).
\url{https://www.diva-portal.org/smash/record.jsf?pid=diva2%3A315832&dswid=5796}
\end{botherref}
\endbibitem

\bibitem[\protect\citeauthoryear{D{\"u}r et~al.}{1999}]{dur1999separability}
\begin{barticle}
\bauthor{\bsnm{D{\"u}r}, \binits{W.}},
\bauthor{\bsnm{Cirac}, \binits{J.I.}},
\bauthor{\bsnm{Tarrach}, \binits{R.}}:
\batitle{Separability and distillability of multiparticle quantum systems}.
\bjtitle{Physical Review Letters}
\bvolume{83}(\bissue{17}),
\bfpage{3562}
(\byear{1999})
\doiurl{10.1103/PhysRevLett.83.3562}
\end{barticle}
\endbibitem

\bibitem[\protect\citeauthoryear{D{\"u}r et~al.}{2000}]{dur2000three}
\begin{barticle}
\bauthor{\bsnm{D{\"u}r}, \binits{W.}},
\bauthor{\bsnm{Vidal}, \binits{G.}},
\bauthor{\bsnm{Cirac}, \binits{J.I.}}:
\batitle{Three qubits can be entangled in two inequivalent ways}.
\bjtitle{Physical Review A}
\bvolume{62}(\bissue{6}),
\bfpage{062314}
(\byear{2000})
\doiurl{10.1103/physreva.62.062314}
\end{barticle}
\endbibitem

\bibitem[\protect\citeauthoryear{G{\"u}hne and Seevinck}{2010}]{guhne2010separability}
\begin{barticle}
\bauthor{\bsnm{G{\"u}hne}, \binits{O.}},
\bauthor{\bsnm{Seevinck}, \binits{M.}}:
\batitle{Separability criteria for genuine multiparticle entanglement}.
\bjtitle{New Journal of Physics}
\bvolume{12}(\bissue{5}),
\bfpage{053002}
(\byear{2010})
\doiurl{10.1088/1367-2630/12/5/053002}
\end{barticle}
\endbibitem

\bibitem[\protect\citeauthoryear{Gühne and Tóth}{2009}]{guhne2009entanglement}
\begin{barticle}
\bauthor{\bsnm{Gühne}, \binits{O.}},
\bauthor{\bsnm{Tóth}, \binits{G.}}:
\batitle{{Entanglement detection}}.
\bjtitle{Phys. Rep.}
\bvolume{474}(\bissue{1}),
\bfpage{1}
(\byear{2009})
\doiurl{10.1016/j.physrep.2009.02.004}
\end{barticle}
\endbibitem

\bibitem[\protect\citeauthoryear{Smith et~al.}{2026}]{smith2026fully}
\begin{barticle}
\bauthor{\bsnm{Smith}, \binits{I.D.}},
\bauthor{\bsnm{Wolfe}, \binits{E.}},
\bauthor{\bsnm{Spekkens}, \binits{R.W.}}:
\batitle{Fully quantum inflation: quantum marginal problem constraints in the service of causal inference}.
\bjtitle{PRX Quantum}
\bvolume{7}(\bissue{1}),
\bfpage{010351}
(\byear{2026})
\doiurl{10.1103/244y-nh5n}
\end{barticle}
\endbibitem

\bibitem[\protect\citeauthoryear{Basu et~al.}{2007}]{basu2007algorithms}
\begin{bbook}
\bauthor{\bsnm{Basu}, \binits{S.}},
\bauthor{\bsnm{Pollack}, \binits{R.}},
\bauthor{\bsnm{Coste-Roy}, \binits{M.F.}}:
\bbtitle{Algorithms in Real Algebraic Geometry}.
\bsertitle{Algorithms and Computation in Mathematics}.
\bpublisher{Springer},
\blocation{Heidelberg, Germany}
(\byear{2007}).
\doiurl{10.1007/3-540-33099-2}
\end{bbook}
\endbibitem

\bibitem[\protect\citeauthoryear{Weyl}{1946}]{weyl1946classical}
\begin{bbook}
\bauthor{\bsnm{Weyl}, \binits{H.}}:
\bbtitle{The Classical Groups: Their Invariants and Representations}
vol. \bseriesno{1}.
\bpublisher{Princeton university press},
\blocation{Princeton, New Jersey}
(\byear{1946}).
\doiurl{10.1515/9781400883905}
\end{bbook}
\endbibitem

\bibitem[\protect\citeauthoryear{Procesi and Schwarz}{1985}]{procesi1985inequalities}
\begin{barticle}
\bauthor{\bsnm{Procesi}, \binits{C.}},
\bauthor{\bsnm{Schwarz}, \binits{G.}}:
\batitle{Inequalities defining orbit spaces}.
\bjtitle{Inventiones mathematicae}
\bvolume{81}(\bissue{3}),
\bfpage{539}--\blpage{554}
(\byear{1985})
\doiurl{10.1007/BF01388587}
\end{barticle}
\endbibitem

\bibitem[\protect\citeauthoryear{Blakaj}{2025}]{VjosaPhDthesis}
\begin{botherref}
\oauthor{\bsnm{Blakaj}, \binits{V.}}:
Transcendental tools in quantum information theory.
PhD thesis,
Technical University of Munich
(2025).
\url{https://nbn-resolving.org/urn:nbn:de:bvb:91-diss-20250326-1764825-0-4}
\end{botherref}
\endbibitem

\bibitem[\protect\citeauthoryear{Horn and Johnson}{1985}]{horn1985matrix}
\begin{bbook}
\bauthor{\bsnm{Horn}, \binits{R.A.}},
\bauthor{\bsnm{Johnson}, \binits{C.R.}}:
\bbtitle{Matrix Analysis}.
\bpublisher{Cambridge university press},
\blocation{Cambridge}
(\byear{1985}).
\doiurl{10.1017/CBO9780511810817}
\end{bbook}
\endbibitem

\bibitem[\protect\citeauthoryear{Hayashi}{2002}]{hayashi2002optimal}
\begin{barticle}
\bauthor{\bsnm{Hayashi}, \binits{M.}}:
\batitle{Optimal sequence of quantum measurements in the sense of {S}tein's lemma in quantum hypothesis testing}.
\bjtitle{Journal of Physics A: Mathematical and General}
\bvolume{35}(\bissue{50}),
\bfpage{10759}--\blpage{10773}
(\byear{2002})
\doiurl{10.1088/0305-4470/35/50/307}
\end{barticle}
\endbibitem

\bibitem[\protect\citeauthoryear{Christandl and Mitchison}{2006}]{christandl2006spectra}
\begin{barticle}
\bauthor{\bsnm{Christandl}, \binits{M.}},
\bauthor{\bsnm{Mitchison}, \binits{G.}}:
\batitle{The spectra of quantum states and the {K}ronecker coefficients of the symmetric group}.
\bjtitle{Communications in Mathematical Physics}
\bvolume{261}(\bissue{3}),
\bfpage{789}--\blpage{797}
(\byear{2006})
\doiurl{10.1007/s00220-005-1435-1}
\end{barticle}
\endbibitem

\bibitem[\protect\citeauthoryear{Keyl and Werner}{2001}]{keyl2001estimating}
\begin{barticle}
\bauthor{\bsnm{Keyl}, \binits{M.}},
\bauthor{\bsnm{Werner}, \binits{R.F.}}:
\batitle{Estimating the spectrum of a density operator}.
\bjtitle{Physical Review A}
\bvolume{64}(\bissue{5}),
\bfpage{052311}
(\byear{2001})
\doiurl{10.1142/9789812563071_0030}
\end{barticle}
\endbibitem

\bibitem[\protect\citeauthoryear{Folland}{2016}]{folland2016course}
\begin{bbook}
\bauthor{\bsnm{Folland}, \binits{G.B.}}:
\bbtitle{A Course in Abstract Harmonic Analysis}.
\bpublisher{CRC press},
\blocation{Boca Raton, FL}
(\byear{2016}).
\doiurl{10.1201/b19172}
\end{bbook}
\endbibitem

\bibitem[\protect\citeauthoryear{Pak and Panova}{2023}]{pak2023durfee}
\begin{barticle}
\bauthor{\bsnm{Pak}, \binits{I.}},
\bauthor{\bsnm{Panova}, \binits{G.}}:
\batitle{Durfee squares, symmetric partitions and bounds on {K}ronecker coefficients}.
\bjtitle{Journal of Algebra}
\bvolume{629},
\bfpage{358}--\blpage{380}
(\byear{2023})
\doiurl{10.1016/j.jalgebra.2023.04.006}
\end{barticle}
\endbibitem

\bibitem[\protect\citeauthoryear{Keyl}{2006}]{keyl2006quantum}
\begin{barticle}
\bauthor{\bsnm{Keyl}, \binits{M.}}:
\batitle{Quantum state estimation and large deviations}.
\bjtitle{Reviews in Mathematical Physics}
\bvolume{18}(\bissue{01}),
\bfpage{19}--\blpage{60}
(\byear{2006})
\doiurl{10.1142/s0129055x06002565}
\end{barticle}
\endbibitem

\bibitem[\protect\citeauthoryear{Hall}{2015}]{hall2015lie}
\begin{bbook}
\bauthor{\bsnm{Hall}, \binits{B.C.}}:
\bbtitle{Lie Groups, Lie Algebras, and Representations: An Elementary Introduction},
\bedition{2}nd edn.
\bsertitle{Graduate Texts in Mathematics 222}.
\bpublisher{Springer},
\blocation{Switzerland}
(\byear{2015}).
\doiurl{10.1007/978-3-319-13467-3}
\end{bbook}
\endbibitem

\bibitem[\protect\citeauthoryear{Zhelobenko}{1973}]{zhelobenko1973compact}
\begin{bbook}
\bauthor{\bsnm{Zhelobenko}, \binits{D.P.}}:
\bbtitle{Compact Lie Groups and Their Representations}
vol. \bseriesno{40}.
\bpublisher{American Mathematical Society},
\blocation{Providence, Rhode Island}
(\byear{1973}).
\doiurl{10.1090/mmono/040}
\end{bbook}
\endbibitem

\bibitem[\protect\citeauthoryear{O'Donnell and Wright}{2016}]{odonnell2016efficient}
\begin{bchapter}
\bauthor{\bsnm{O'Donnell}, \binits{R.}},
\bauthor{\bsnm{Wright}, \binits{J.}}:
\bctitle{Efficient quantum tomography}.
In: \bbtitle{Proceedings of the Forty-eighth Annual ACM Symposium on Theory of Computing},
pp. \bfpage{899}--\blpage{912}
(\byear{2016}).
\doiurl{10.1145/2897518.2897544}
\end{bchapter}
\endbibitem

\bibitem[\protect\citeauthoryear{Audenaert and Datta}{2015}]{audenaert2013alpha}
\begin{barticle}
\bauthor{\bsnm{Audenaert}, \binits{K.M.}},
\bauthor{\bsnm{Datta}, \binits{N.}}:
\batitle{Alpha-z-relative {R}ényi entropies}.
\bjtitle{Journal of Mathematical Physics}
\bvolume{56},
\bfpage{022202}
(\byear{2015})
\doiurl{10.1063/1.4906367}
\end{barticle}
\endbibitem

\bibitem[\protect\citeauthoryear{Franks and Walter}{2020}]{franks2020minimal}
\begin{barticle}
\bauthor{\bsnm{Franks}, \binits{C.}},
\bauthor{\bsnm{Walter}, \binits{M.}}:
\batitle{Minimal length in an orbit closure as a semiclassical limit}.
\bjtitle{arXiv:2004.14872v2}
(\byear{2020})
\doiurl{10.48550/arXiv.2004.14872}
\end{barticle}
\endbibitem

\bibitem[\protect\citeauthoryear{Botero et~al.}{2021}]{botero2021large}
\begin{barticle}
\bauthor{\bsnm{Botero}, \binits{A.}},
\bauthor{\bsnm{Christandl}, \binits{M.}},
\bauthor{\bsnm{Vrana}, \binits{P.}}:
\batitle{Large deviation principle for moment map estimation}.
\bjtitle{Electronic Journal of Probability}
\bvolume{26},
\bfpage{1}--\blpage{23}
(\byear{2021})
\doiurl{10.1214/21-ejp636}
\end{barticle}
\endbibitem

\bibitem[\protect\citeauthoryear{Fraser}{2023}]{fraser2023estimation}
\begin{botherref}
\oauthor{\bsnm{Fraser}, \binits{T.C.}}:
An estimation theoretic approach to quantum realizability problems.
PhD thesis,
University of Waterloo
(2023).
\url{http://hdl.handle.net/10012/19970}
\end{botherref}
\endbibitem

\bibitem[\protect\citeauthoryear{Rubboli et~al.}{2024}]{rubboli2024mixed}
\begin{barticle}
\bauthor{\bsnm{Rubboli}, \binits{R.}},
\bauthor{\bsnm{Takagi}, \binits{R.}},
\bauthor{\bsnm{Tomamichel}, \binits{M.}}:
\batitle{Mixed-state additivity properties of magic monotones based on quantum relative entropies for single-qubit states and beyond}.
\bjtitle{Quantum}
\bvolume{8},
\bfpage{1492}
(\byear{2024})
\doiurl{10.22331/q-2024-10-04-1492}
\end{barticle}
\endbibitem

\end{thebibliography}

\begin{appendices}
\crefalias{section}{appendix}

\section{Semialgebraicity of separable states}
\label{sec:semialgebraic}

In \cref{sec:discussion}, it was mentioned that the set of bipartite separable states for fixed local dimensions, is semialgebraic.
In other words, $\SEPAB$ can be described by finitely 
many polynomial (in)equalities in the coefficients of $\rho_{AB}$, viewed as element of the real vector space of local dimension $(d_{A}, d_{B})$ Hermitian matrices, i.e., $\SEPAB$ can be characterized using finitely many polynomial entanglement witnesses.
This property of $\SEPAB$ follows from a combination of Sylvester's criterion (\cref{sec:sylvester}),  Carathéodory's theorem, and the Tarski-Seidenberg theorem \cite{basu2007algorithms}.
In this section, we further show that $\SEPAB$ admits a semialgebraic description in terms of finitely many local-unitary invariant polynomial entanglement witnesses.

The orbit space 
\begin{equation}
    \widetilde{\SEP}(A\colonsep B) \coloneq \SEPAB/\bigl(U(d_{A})\times U(d_{B})\bigr),
\end{equation}
(equipped with the quotient topology), is an \textit{abstract} quotient, and thus it is not \textit{a priori} a subset of some $\mathbb R^N$. 
To say that $\widetilde{\SEP}(A\colonsep B)$ \textit{admits a semialgebraic model}
means that there exist an $m \in \mathbb{N}$ and a semialgebraic set $Y\subseteq \mathbb R^m$ together with a homeomorphism, such that  $\widetilde{\SEP}(A\colonsep B)\ \cong\ Y$.
Although we can show that such a finite characterization exists, we are currently unable 
to provide such a finite characterization of $\SEPAB$ for all local dimensions $(d_{A}, d_{B})$ where $d_{A}d_{B} > 6$\footnote{For local dimensions $(d_{A}, d_{B})$ where $d_{A}d_{B} \leq 6$, the PPT criterion is necessary and sufficient \cite{horodecki2001separability,elben2020mixed}.}.

Let us start by recalling fundamental results concerning the action of a compact Lie group $K$, on a finite dimensional real vector spaces $W$.
In our setting, we consider the local unitary group $K = U(d_{A})\times U(d_{B})$, acting on the real vector space $W=\Herm(A\colonsep B)$ of Hermitian matrices on $\mathbb C^{d_{A}} \otimes \mathbb C^{d_{B}}$ by conjugation, i.e., for all $(U, V) \in U(d_{A}) \times U(d_{B})$ and $\rho_{AB} \in \Herm(A\colonsep B)$
\begin{equation}\label{eq:action}
    \rho_{AB} \mapsto (U\otimes V) \rho_{AB} (U\otimes V)^\dagger.
\end{equation}
Note that this action is orthogonal with respect to the Hilbert-Schmidt norm.

In general, for any real vector space $W$ and group $K$ by $\mathbb R[W]$ is denoted the ring of real polynomial functions on $W$, and by $\mathbb R[W]^K$ the subring of $K$-invariant polynomial functions on $W$:
\begin{equation}
    \mathbb R[W]^K = \{f\in \mathbb R[W] \mid \forall k \in K, w \in W: f(k\cdot w)=f(w)\}.
\end{equation}
By the results of Hilbert and Hurwitz (see \cite[Theorem 8.14.A, p. 274]{weyl1946classical}) the ring $\mathbb R[W]^K$ is \textit{finitely} generated when $K$ is compact.
Let $p_1,\dots,p_m \in \mathbb R[W]^K$ be a finite list of generators. 
We define the  map $P : W \to \mathbb R^{m}$ by
\begin{align}
    P(w) = (p_1(w),\dots,p_m(w)).
\end{align}
By the Tarski-Seidenberg theorem \cite{basu2007algorithms}, the image $X \coloneq \mathrm{im}(P) \subseteq \mathbb R^{m}$ of the polynomial map $P$ is semialgebraic and homeomorphic to the quotient $W/K$ by the following proposition.
\begin{prop}[\cite{procesi1985inequalities}, Proposition 0.4]
    \label{prop:PS}
    The map $P : W \to \mathbb R^{m}$ is proper and induces a homeomorphism $\bar P : W/K \to X$ between the quotient $W/K$ and its image $X = \mathrm{im}(P) \subseteq \mathbb R^{m}$.
\end{prop}
Moreover, the image of $P$ restricted to any semialgebraic subset $S$ of $W$ is also semialgebraic by the following.
\begin{prop}[\cite{VjosaPhDthesis}, Chapter 4]
    \label{prop:image_semialg} 
    If $S\subseteq W$ is a semialgebraic subset of $W$, then its image, $P(S) = \mathrm{im}(P|_{S})$, under a polynomial map $P : W \to \mathbb R^{m}$ is also semialgebraic.
\end{prop}
Building on the above, we prove the following.
\begin{thm}
    \label{thm:orbit_semialg}
    If $S \subseteq W$ is a semialgebraic subset of $W$ that is $K$-invariant, then the orbit space $S/K$ is homeomorphic to the semialgebraic set $P(S) \subseteq \mathbb R^m$.
\end{thm}
\begin{proof}
    Let $\pi_S : S\to S/K$ be the orbit map restricted to $S$, i.e., for $x \in S$ we have $\pi_S(x)=K \cdot x = [x]$.
    Since each component polynomial of $P=(p_1,\ldots, p_m)$ is $K$-invariant on $W$, $P$ is constant on each $K$-orbit $K\cdot x = [x]$ of an element $x \in S$.
    From here we obtain a well-defined continuous map $\bar P_S : S/K \to Y$ that is the restriction of the homeomorphism $\bar P$ from \cref{prop:PS} and $Y = P(S) = \mathrm{im}(P|_{S}) = \mathrm{im}(\bar P_S) = \mathrm{im}(\bar P|_{S/K})$. Finally, 
    by \cref{prop:image_semialg} the set $Y$ is semialgebraic, which proves the claim.
\end{proof}

In our setting, the set $S = \SEPAB$ is a $K=U(d_{A}) \times U(d_{B})$-invariant semialgebraic subset of $W = \Herm(A\colonsep B)$. Therefore, \cref{thm:orbit_semialg} yields the following result.
\begin{cor}
    \label{cor:sep_orbit_semialg}
    The orbit space
    \begin{equation}
        \widetilde{\SEP}(A\colonsep B) = \SEPAB/(U(d_{A}) \times U(d_{B}))
    \end{equation}
    is homeomorphic to the semialgebraic set $P(\SEPAB) \subseteq \mathbb R^m$.
    In particular, $\SEPAB$ admits a semialgebraic model consisting entirely of $U(d_{A})\times U(d_{B})$-invariant polynomials.
\end{cor}

\section{Sylvester's criterion}
\label{sec:sylvester}
A necessary and sufficient condition for the positive semidefiniteness of a Hermitian matrix is given by the non-negativity of all its principal minors; this is known as Sylvester’s criterion \cite[Theorem 7.2.5 (a)]{horn1985matrix}. 
In what follows, we present a well-known variation of this criterion, which we later use to convert operator inequalities into polynomial entanglement witnesses.
\begin{prop}
    \label{prop:sylvesters}
    Let $X$ be an $\ell \times \ell$ Hermitian matrix acting on a Hilbert space $L = \mathbb C^{\ell}$.
    Then, $X \succeq 0$ if and only if for all $m \in \{1, \ldots, \ell\}$
    \begin{equation}
        E_{m}(X) \geq 0,
    \end{equation}
    where $E_{m}(X) = \Tr(P_{\wedge^{m}L} X^{\ot m})$ and $P_{\wedge^{m}L}$ is the projection operator onto the skew-symmetric subspace $\wedge^{m}L \subseteq L^{\ot m}$. 
\end{prop}
Note that $E_{m}(X)$ is also equivalent to both $(i)$ the sum over all principal minors of $X$ of degree $m$, and $(ii)$ the elementary symmetric polynomial in the eigenvalues of $X$.
We now provide the proof of \cref{prop:sylvesters}.
\begin{proof}[Proof of \cref{prop:sylvesters}]
    The characteristic polynomial of $X$,
    \begin{equation}
        p_X(t) = \det(X-tI) = \Tr(P_{\wedge^{\ell} L} (X-tI)^{\ot \ell}),
    \end{equation}
    is related to the quantities $E_m(X)$ for $m \in \{1,\ldots, \ell\}$ by the identity
    \begin{align}
    \begin{split}
        p_X(t) 
        &= (-t)^{\ell} + E_1(X)(-t)^{\ell-1} + \cdots + E_{\ell-1}(X) (-t) + E_\ell(X). 
    \end{split}
    \end{align}
    Since the eigenvalues of $X$ are the roots of $p_X(t)$, the number of strictly negative eigenvalues of $X$ is equal to the number of strictly positive roots of the polynomial $q_X(t) = p_X(-t)$ (both counted with multiplicities), where
    \begin{equation}
        q_X(t) = t^{\ell} + E_1(X)t^{\ell-1} + \cdots + E_{\ell-1}(X) t + E_\ell(X). 
    \end{equation}
    The assumption that $E_m(X) \geq 0$ for all $m \in \{1,\ldots,\ell\}$ implies that all coefficients of $q_X(t)$ are non-negative.
    Therefore, the number of strictly positive roots of $q_{X}(t)$ is zero as $q_X(t) > 0$, whenever $t > 0$.  
    This establishes that if $E_{m}(X) \geq 0$ for all $m \in \{1,\ldots, \ell\}$, then $X \succeq 0$.
    The converse statement follows because $X \succeq 0$ implies $X^{\otimes m} \succeq 0$, and hence $E_m(X) \geq 0$, as $P_{\wedge^{m}L}$ is positive semidefinite.
\end{proof}

In the main text, we applied \cref{prop:sylvesters} to the positive semidefinite operator $X = \Lambda_{A^{n}B^{n}} - \rho_{AB}^{\ot n}$ from \cref{thm:stronger_forward} to obtain an polynomial entanglement witnesses $W_{A^{nm}B^{nm}}^{n,m}$ for $\rho_{AB}$ as in \cref{eq:sylvester_to_witness}.
The following lemma (after making the substitutions $\ell=d_{AB}^{n}$, $\omega = \rho_{AB}^{\otimes n}$ and $\Lambda = \Lambda_{A^{n}B^{n}}$) provides a formula for computing the witness operator $W_{A^{nm}B^{nm}}^{n,m}$ from $\Lambda_{A^{n}B^{n}}$ for each value of $m$ using
\begin{equation}
    W_{A^{nm}B^{nm}}^{n,m} \coloneqq \Gamma_m[\Lambda_{A^{n}B^{n}}].
\end{equation}
\begin{lem}
    \label{lem:witness_from_lambda}
    Let $\Lambda, \omega$ be $\ell \times \ell$ Hermitian matrices with $\Tr(\omega) = 1$.
    Then $\omega \preceq \Lambda$ if and only if, for all $m \in \{1, \ldots, \ell\}$
    \begin{equation}
        \Tr[\Gamma_m[\Lambda] \omega^{\ot m}] \geq 0
    \end{equation}
    where $\Gamma_m[\Lambda]$ is the Hermitian operator on $L^{\ot m}$ defined by
    \begin{equation}
        \label{eq:witness_formula}
        \Gamma_m[\Lambda] = \sum_{j=0}^{m} (-1)^{j}\binom{m}{j}  \Tr_{1,\ldots,m-j}[P_{\wedge^{m}L} (\Lambda^{\ot m-j}\ot I^{\ot j} )] \ot I^{\ot m-j}.
    \end{equation}
\end{lem}
\begin{proof}
    From Sylvester's criterion we have $\omega \preceq \Lambda$ if and only if for all $m \in \{1, \ldots, \ell\}$
    \begin{equation}
        E_m(\Lambda - \omega) =\Tr(P_{\wedge^{m}L} (\Lambda - \omega)^{\ot m}) \geq 0.
    \end{equation}
    Now we apply the binomial expansion to $(\Lambda - \omega)^{\ot m}$ and note that $P_{\wedge^{m}L}$ is invariant under conjugation by permutations in $S_m$ so we obtain
    \begin{align}
         E_m(\Lambda - \omega)
         &= \sum_{j=0}^{m} (-1)^{j} \binom{m}{j} \Tr(P_{\wedge^{m}L} \Lambda^{\ot m-j} \ot \omega^{\ot j}).
    \end{align}
    From here we see that the definition of $\Gamma_m$ in \cref{eq:witness_formula} ensures 
    \begin{equation}
        E_m(\Lambda - \omega) = \Tr(\Gamma_m[\Lambda] \omega^{\ot m}),
    \end{equation}
    provided that $\Tr(\omega) = 1$.
\end{proof}

\section{Twirled pinching inequalities}
In this section, we establishes an operator inequality relating a positive semidefinite operator $X$ to its image $\twirl_{G}(X)$ under \textit{twirling}, or equivalently \textit{pinching}, by a compact group $G$.

Throughout, let $\pi : G \to \mathrm{GL}(H)$ be a unitary representation of $G$ on a finite-dimensional complex vector space $H = \mathbb C^{D}$. We assume that $H$ decomposes into irreducible subrepresentations $V_{\lambda}$ as
\begin{equation}
    \label{eq:mashcke_decomp}
    H = \bigoplus_{\lambda \in \hat G} V_{\lambda} \otimes M_{\lambda},
\end{equation}
with $M_{\lambda}$ denoting the corresponding multiplicity spaces.

Let $\twirl_{G}$ denote the twirling map with respect to $G$:
\begin{align}
    \twirl_{G}(X) 
    &= \int d \mu(g) \pi(g) X \pi(g^{-1}) , \\
    &= \bigoplus_{\lambda \in \hat G} \frac{I_{V_{\lambda}}}{d_{\lambda}} \otimes \Tr_{V_{\lambda}}(P_{\lambda} X P_{\lambda}).
\end{align}
Here $d_{\lambda} = \dim V_{\lambda}$ and $P_{\lambda}$ denotes the projection onto the $\lambda$-isotypic subspace $V_{\lambda} \otimes M_{\lambda} \subseteq H$.

We now prove an operator bound between any positive semidefinite operator $X \succeq 0$ and its image under twirling $\twirl_G(X)$. This bound is related to Hayashi's pinching inequality \cite[Lemma 9]{hayashi2002optimal}, adapted to the subspaces of $H$ in \cref{eq:mashcke_decomp}.
\begin{thm}
    \label{thm:pinching}
    For every positive semidefinite operator $X \succeq 0$, we have
    \begin{equation}
        X \preceq \Gamma_{G} \twirl_G(X),
    \end{equation}
    where  
    \begin{equation}
        \Gamma_{G} = \sum_{\lambda \in \hat G} d_{\lambda} \min\{d_\lambda, m_\lambda\},
    \end{equation}
    with $d_{\lambda} = \dim V_\lambda$ and $m_{\lambda} = \dim M_{\lambda}$.
\end{thm}
\begin{proof}
    It suffices to prove the inequality for rank-one operators $X = \ketbra{\Psi}$, since the general case follows by spectral decomposition: if $X = \sum_{i}x_i \ketbra{\Psi_i}$ and $x_i\ge 0$ then the desired inequality for each $\ketbra{\Psi_i}$ implies it for $X$.
    Let $\ket{\Psi} = \bigoplus_{\lambda} \ket{\psi_{\lambda}}$ be the decomposition of $\ket{\Psi} \in H$ into isotypic blocks, where $\ket{\psi_{\lambda}} \in V_{\lambda} \otimes M_{\lambda}$ and note that
    \begin{align}
        \twirl_G(\ketbra{\Psi}) 
        &= \bigoplus_{\lambda \in \hat G} \frac{I_{V_{\lambda}}}{d_{\lambda}} \otimes \Tr_{V_{\lambda}}(\ketbra{\psi_{\lambda}}).
    \end{align}
    Now, for any $\ket{\psi_{\lambda}} \in V_{\lambda} \otimes M_{\lambda}$ with Schmidt rank $s \leq \min\{d_{\lambda}, m_{\lambda}\}$, one always has
    \begin{equation}
        \ketbra{\psi_{\lambda}} \preceq s I_{V_{\lambda}} \ot \Tr_{V_{\lambda}}(\ketbra{\psi_{\lambda}}).
    \end{equation}
   Hence
    \begin{equation}
        \ketbra{\psi_{\lambda}} \preceq \gamma_{\lambda} \left(\frac{I_{V_{\lambda}}}{d_{\lambda}} \ot \Tr_{V_{\lambda}}(\ketbra{\psi_{\lambda}})\right),
    \end{equation}
    where $\gamma_{\lambda} \coloneq d_{\lambda} \min \{d_{\lambda}, m_{\lambda}\}$.
    Let $\ket{\Phi} \in H$ be an arbitrary vector with block decomposition $\ket{\Phi} = \bigoplus_{\lambda} \ket{\phi_{\lambda}}$.
    Applying the Cauchy--Schwarz inequality, we obtain
    \begin{align}
        \abs{\braket{\Phi, \Psi}}^{2}
        &= \abs{\sum_{\lambda} \braket{\phi_\lambda, \psi_\lambda}}^2 \\
        &= \abs{\sum_{\lambda} \sqrt\gamma_{\lambda}\frac{\braket{\phi_\lambda, \psi_\lambda}}{\sqrt\gamma_\lambda}}^2 \\\
        &\leq \left(\sum_{\mu} \gamma_{\mu}\right) \left(\sum_{\lambda}\frac{\abs{\braket{\phi_\lambda, \psi_\lambda}}^{2}}{\gamma_\lambda}\right) \\
        &\leq \left(\sum_{\mu} \gamma_{\mu}\right) \sum_{\lambda}\bra{\Phi}\left(\frac{I_{V_{\lambda}}}{d_{\lambda}} \ot \Tr_{V_{\lambda}}(\ketbra{\psi_{\lambda}})\right)\ket{\Phi} \\
        &= \Gamma_{G} \bra{\Phi} \twirl_G(\ketbra{\Psi}) \ket{\Phi}.
    \end{align}
    This proves the claim for $X = \ketbra{\Psi}$, and hence for all $X \succeq 0$.
\end{proof}

\section{The local unitary Fourier transform}
\label{sec:fourier_transform}
We now exploit the symmetries of our operator inequalities to block-diagonalize the characterization of \cref{thm:stronger_forward}. This yields an equivalent formulation of separability in terms of the Fourier components of $\rho_{AB}$ with respect to the local unitary group $G = U(d_{A}) \times U(d_{B})$.
To begin, we recall Schur--Weyl duality asserts that the Hilbert space $H^{\otimes n} = (\mathbb C^{d})^{\otimes n}$, viewed as a representation of $S_n\times U(d)$, can be decomposed into sectors indexed by (sorted) integer partitions $\lambda = (\lambda_1\ldots, \lambda_{d})$ of $n$, with at most $d$ nonzero parts:
\begin{equation}
    \label{eq:schur_weyl}
    H^{\otimes n} 
    = 
    \bigoplus_{\lambda \vdash n} V_{\lambda} \ot U_{\lambda}^{d}.
\end{equation}
Here $V_{\lambda}$ carries an irreducible representation of the symmetric group $S_n$, while $U_{\lambda}^{d}$ carries an irreducible representation of the unitary group $U(d)$ where $d = \dim(H)$. 
Implicit in this decomposition is the condition that $\lambda$ has at most $d$ nonzero parts, since otherwise $U_{\lambda}^{d} = \{0\}$.

With respect to this decomposition, the $S_n$-twirling channel acts on an operator $X$ on $H^{\otimes n}$ as
\begin{align}
    \twirl_{S_n}(X) = \bigoplus_{\lambda \vdash n} \frac{I_{V_{\lambda}}}{\dim V_{\lambda}} \ot \Tr_{V_{\lambda}} [ P_{\lambda} X P_{\lambda} ],
\end{align}
where $P_{\lambda}$ denotes the orthogonal projection onto the subspace $V_{\lambda} \ot U_{\lambda}^{d} \subseteq (\mathbb C^{d})^{\otimes n}$.
Similarly, the bipartite space $(A\otimes B)^{\otimes n}$ decomposes, as a representation of $S_n \times U(d_{A}) \times U(d_{B})$, as
\begin{equation}
    (A \otimes B)^{\otimes n} = \bigoplus_{\alpha, \beta, \lambda \vdash n} V_{\lambda} \ot U_{\alpha}^{a} \otimes U_{\beta}^{b} \otimes M_{\alpha\beta}^{\lambda}.
\end{equation}
Here $\alpha, \beta$, and $\lambda$ are (sorted) integer partitions of $n$ with at most $d_{A}, d_{B}$, and $d_{AB}$ nonzero parts, respectively.
The multiplicity space is given by
\begin{equation}
    M_{\alpha\beta}^{\lambda} = \mathrm{Hom}_{U(d_{A}) \times U(d_{B})}(U_{\alpha}^{d_{A}} \otimes U_{\beta}^{d_{B}}, U_{\lambda}^{d_{AB}}),
\end{equation}
and its dimension
\begin{equation}
    g_{\alpha \beta}^{\lambda} \coloneqq \dim M_{\alpha\beta}^{\lambda}
\end{equation}
is the corresponding \textit{Kronecker coefficient}. 

Given an operator $X$ on $(A\otimes B)^{\otimes n}$, twirling with respect to the group $S_n \times U(d_{A}) \times U(d_{B})$ gives
\begin{align}
\begin{split}
    \twirl_{S_n\times U(d_{A}) \times U(d_{B})}(X)
    &= \bigoplus_{\lambda,\alpha,\beta \vdash n}\frac{I_{V_{\lambda}}}{\dim V_{\lambda}} \ot \frac{I_{U_{\alpha}^{d_{A}}}}{\dim U_{\alpha}^{d_{A}}} \ot \frac{I_{U_{\beta}^{d_{B}}}}{\dim U_{\beta}^{d_{B}}} \ot K_{\alpha\beta}^{\lambda}(X),
\end{split}
\end{align}
where $K_{\alpha\beta}^{\lambda} (X)$ is the operator on $M_{\alpha\beta}^{\lambda}$ defined by
\begin{equation}
    K_{\alpha\beta}^{\lambda} (X) \coloneqq \Tr_{V_{\lambda} \ot U_{\alpha}^{d_{A}} \otimes U_{\beta}^{d_{B}}} [ P_{\alpha\beta}^{\lambda} X P_{\alpha\beta}^{\lambda} ].
\end{equation}
Here, $P_{\alpha\beta}^{\lambda} = P_{\lambda}(P_{\alpha} \ot P_{\beta})$ denotes the orthogonal projection onto the $(\alpha,\beta,\lambda)$-isotypic subspace of $(A\otimes B)^{\otimes n}$.
By construction, $K_{\alpha\beta}^{\lambda}$ is $S_n \times U(d_{A}) \times U(d_{B})$-invariant, in the sense that
\begin{equation}
    K_{\alpha\beta}^{\lambda} \circ \twirl_{S_n\times U(d_{A}) \times U(d_{B})} = K_{\alpha\beta}^{\lambda}.
\end{equation}
The trace of $K_{\alpha\beta}^{\lambda} (\rho_{AB}^{\otimes n})$, namely
\begin{equation}
    q_\rho(\alpha,\beta,\lambda) \coloneqq \Tr(P_{\alpha\beta}^{\lambda}\rho_{AB}^{\otimes n}),
\end{equation}
defines a probability distribution over triples of partitions $\alpha,\beta,\lambda \vdash n$.
This distribution was studied by Christandl \& Mitchison \cite{christandl2006spectra}, who showed, using spectrum-estimation techniques \cite{keyl2001estimating}, that it concentrates around the spectra of $\rho_{AB}$ and of its marginals $\rho_{A}$ and $\rho_{B}$.

More generally, by the Peter--Weyl theorem \cite[Chapter 5]{folland2016course}, every polynomial function of $\rho_{AB}$ that is invariant under $(U(d_{A})\times U(d_{B}))$ can be expanded as a linear combination of matrix coefficients of the operators $K_{\alpha\beta}^{\lambda} (\rho_{AB}^{\otimes n})$, as $\alpha,\beta, \lambda$ range over partitions of $n \in \mathbb N$. In this sense, the operators $K_{\alpha\beta}^{\lambda} (\rho_{AB}^{\otimes n})$, together with their matrix coefficients, are precisely the Fourier components of $\rho_{AB}$ with respect to the non-abelian local-unitary Fourier transform. 

In this language, the Fourier-transform version of our main result takes the following form.
\begin{cor}
    \label{cor:fourier_transform_ineq}
    A state $\rho_{AB}$ belongs to $ \SEPAB$ $\Longleftrightarrow$ for all $n$ and all $\alpha,\beta,\lambda \vdash n$,
    \begin{equation}
        K_{\alpha\beta}^{\lambda}(\rho_{AB}^{\otimes n}) \preceq K_{\alpha\beta}^{\lambda}(\Lambda_{A^{n}B^{n}}).
    \end{equation}
\end{cor}
\begin{proof}
    \textbf{(}$\boldsymbol{\Longrightarrow}$\textbf{)}
    The forward implication follows from the $U(d_{A}) \times U(d_{B})$-invariance of $\SEPAB$, together with \cref{thm:stronger_forward}. Indeed,
    by definition,
    \begin{equation}
        \twirl_{U(d_{A}) \times U(d_{B})}(\rho_{AB}^{\otimes n}) = \iint d U d V [(U \otimes V) \rho_{AB} (U^{\dagger} \otimes V^{\dagger})]^{\otimes n}.
    \end{equation}
    Since every state appearing in the integrand above belongs to $\SEPAB$, the forward direction of \cref{thm:stronger_forward} applies pointwise, and the claimed inequality follows.

    \textbf{(}$\boldsymbol{\Longleftarrow}$\textbf{)}
    The reverse implication is not immediate from  \cref{thm:completeness}, since in general $\rho_{AB}^{\otimes n} \neq \twirl_{U(d_{A}) \times U(d_{B})}(\rho_{AB}^{\otimes n})$. 
    To address this, we apply the pinching inequality from \cref{thm:pinching}, yielding
    \begin{equation}
        \rho_{AB}^{\otimes n} = \twirl_{S_n}(\rho_{AB}^{\otimes n}) \preceq \Gamma(n) \twirl_{U(d_{A}) \times U(d_{B})}(\rho_{AB}^{\otimes n}),
    \end{equation}
    where $\Gamma_{U(d_{A})\times U(d_{B})}(n) \in \mathrm{poly}(n)$ may be chosen as
    \begin{equation}
        \Gamma_{U(d_{A})\times U(d_{B})}(n) = \sum_{\alpha, \beta \vdash n} \dim(U_{\alpha}^{d_{A}})^2\dim(U_{\beta}^{d_{B}})^2.
    \end{equation}
    Finally, the proof of the converse direction in \cref{thm:completeness} is unaffected by multiplication with polynomial factors in $n$. 
    The same argument therefore applies here, with $f_{AB}(n)$ replaced by $f_{AB}(n) \Gamma_{U(d_{A})\times U(d_{B})}(n)$.
\end{proof}
The advantage of \cref{cor:fourier_transform_ineq} over \cref{thm:forward} is that it involves only operators acting on the multiplicity spaces $M_{\alpha\beta}^{\lambda}$, whose dimensions grow polynomially in $n$ when the numbers of parts of $\alpha,\beta$, and $\lambda$ are bounded \cite{pak2023durfee}. By contrast, \cref{thm:forward} involves operators on $(A\otimes B)^{\otimes n}$, whose dimension grows exponentially with $n$.
This reduction in complexity essentially mimics the symmetry-reductions from exponentially-sized to polynomially-sized matrices, required to perform algorithms for entanglement detection based on hierarchies of semidefinite programs \cite{doherty2004complete,harrow2017improved,pena2025tailored}.

\section{Quantum state estimation}
\label{sec:state_estimation}

The goal of this section is to provide a proof of \cref{lem:likelihood_ratio_fidelity} which relies on the covariant quantum state estimation scheme introduced by Keyl \cite{keyl2006quantum}.
Although this result was already established in \cite[Appendix J, Theorem 15]{fraser2022sufficient}, we reproduce the proof details here for completeness.
We being by reviewing the essential representation-theoretic features of this estimation scheme before turning to the proof of \cref{lem:likelihood_ratio_fidelity}.

Let $H = \mathbb C^{d}$ be a $d$-dimensional complex Hilbert space equipped with its standard inner product. Its $n$-th tensor power $H^{\otimes n}$ supports a representation of $U(d) \times S_n$: the symmetric group $S_n$ acts by permuting the tensor factors, as in \cref{eq:tensor_perm}), while the unitary group $U(d)$ acts diagonally by $U^{\otimes n}$.

Schur-Weyl duality states that $H^{\otimes n}$ decomposes into irreducible subrepresentations indexed by sorted integer partitions $\lambda = (\lambda_1 \geq \ldots \geq \lambda_d)$ of $n$, or Young diagrams, with at most $d$ parts, denoted by $\lambda \vdash_{d} n$:
\begin{equation}
    H^{\otimes n} = \bigoplus_{\lambda \vdash_{d} n} V_{\lambda} \ot U_{\lambda}^{d}.
\end{equation}
Here, $V_{\lambda}$ carries an irreducible representation of the symmetric group $S_n$, while $U_{\lambda}^{d}$ carries the irreducible representation $\pi_{\lambda}$ of the unitary group $U(d)$ with highest weight $\lambda$ \cite{hall2015lie}.
In particular, there exists a unique normalized highest-weight vector $\ket{\phi_{\lambda}} \in U_{\lambda}^{d}$ satisfying
\begin{equation}
    \pi_{\lambda}(\mathrm{diag}(x_1, \ldots, x_d)) \ket{\phi_{\lambda}} 
    = x_1^{\lambda_1} \cdots x_d^{\lambda_d} \ket{\phi_{\lambda}}.
\end{equation}
Following the notation of \cite{fraser2022sufficient}, for each sorted integer partition $\lambda \vdash_{d} n$ and each unitary $U \in U(d)$, let $\Phi_{\lambda}^{U}$ denote the projection operator on $H^{\otimes n}$ given by
\begin{equation}
    \label{eq:coherent_projectors}
    \Phi_{\lambda}^{U} = I_{V_{\lambda}} \otimes \ketbra{\phi_{\lambda}^{U}}, 
\end{equation}
where $\ket{\phi_{\lambda}^{U}}$ is the normalized $U$-twirled highest-weight vector $\ket{\phi_{\lambda}^{U}} = \pi_{\lambda}(U)\ket{\phi_{\lambda}}$.
For any state $\sigma$, we have the identity \cite[Eq. (141)]{keyl2006quantum}, proven in \cite[Section 49]{zhelobenko1973compact},
\begin{equation}
    \label{eq:gen_pow_func_identity}
    \Tr(\Phi_{\lambda}^{U} \sigma^{\otimes n}) 
    = \dim(V_{\lambda}) \Delta_{\lambda}(U^{\dagger}\sigma U).
\end{equation}
Here, for any vector $r = (r_1 \geq r_2 \geq \cdots \geq r_d) \in \mathbb R^{d}_{\geq 0}$, the function $\Delta_{r}(X)$ is the \textit{generalized power function} \cite{odonnell2016efficient}, defined by
\begin{equation}
    \Delta_{r}(X) = \prod_{j=1}^{d} \mathrm{pm}_j(X)^{r_j-r_{j+1}},
\end{equation}
with $\mathrm{pm}_j(X)$ denoting the leading $j\times j$ principal minor of $X$.
The notational convention is that $r_{d+1} = 0$.
See also the approach based on the \textit{Weyl-trick} in \cite[Section 3]{audenaert2013alpha}.

Moreover, for each $\lambda\vdash_d n$ and unitary $U \in U(d)$, we let $\eta_{\lambda}^{U}$ denote the density matrix with rational eigenvalues
\begin{equation}
    \label{eq:rational_approx_rho}
    \eta_{\lambda}^{U} = U \mathrm{diag}\left(\frac{\lambda_1}{n},\ldots,\frac{\lambda_d}{n}\right)U^{\dagger}.
\end{equation}
Using this notation, the state-estimation POVM introduced in \cite{keyl2006quantum} assigns to each measurable subset $C \subseteq \mathcal S(H)$ the operator
\begin{equation}
    E_n(C) = \sum_{\lambda \vdash_d n} \dim(U_\lambda^{d}) \int_{U(d)} d U \delta_{C}(\eta_{\lambda}^{U}) \Phi_{\lambda}^{U},
\end{equation}
where $dU$ denotes the Haar measure on $U(d)$ and $\delta_{C}$ is the indicator function of $C$, that is, $\delta_{C}(\rho) = 1$ if $\rho \in C$ and $\delta_{C}(\rho) = 0$ otherwise.

In \cite{keyl2006quantum}, it was shown that for each state $\sigma$, the probability measure $\Tr[E_n(d \rho) \sigma^{\otimes n}]$ on $\mathcal S(H)$ satisfies the large deviation principle with rate function $\rho \mapsto D^{*}(\rho\|\sigma)$, where
\begin{equation}
    \label{eq:keyl_rel_ent}
    D^{*}(\rho\|\sigma) \coloneqq \sum_{j=1}^{d} [x_j \ln x_j -(x_j-x_{j+1}) \log \mathrm{pm}_j(U^{\dagger} \sigma U)].
\end{equation}
Here, $x_1 \geq \cdots \geq x_d$ are the eigenvalues of $\rho$, with the convention $x_{d+1} = 0$, and $U$ is a unitary diagonalizing $\rho$, so that $\rho = U \mathrm{diag}(x_1, \ldots, x_d) U^{\dagger}$  \cite{franks2020minimal,botero2021large,fraser2022sufficient,fraser2023estimation}.\footnote{In \cite{fraser2022sufficient}, the quantity $D^{*}(\rho\|\sigma)$ was referred to as the \textit{Keyl-divergence} and denoted by $K(\rho\|\sigma)$.}
We use the convention $0\log 0 = 0$ in \cref{eq:keyl_rel_ent}, as in \cite{keyl2006quantum}.
Finally, note the identities
\begin{equation}
    D^{*}(\rho\|\sigma) = -\log \left(\frac{\Delta_{x}(U^{\dagger} \sigma U)}{\Delta_{x}(\mathrm{diag(x)})}\right), 
    \qquad \text{and} \qquad \log \Delta_{x}(\mathrm{diag(x)}) = - H(x).
\end{equation}
Given states $\rho$ and $\sigma$, the quantity $D^{*}(\rho\|\sigma)$ is non-negative, with equality if and only if $\rho = \sigma$  \cite{keyl2006quantum}.
Although $D^{*}(\rho\|\sigma)$ behaves similarly to the quantum relative entropy $D(\rho \| \sigma) = \Tr(\rho (\log \rho - \log \sigma))$, the two quantities are distinct. In particular, 
\begin{equation}
    D^{*}(\rho\|\sigma) \leq D(\rho \| \sigma).
\end{equation}

In \cite[Appendix J, Theorem 15]{fraser2022sufficient}, this machinery was used to prove the following result. 
\begin{lem}
    \label{lem:likelihood_ratio_keyl_upper_bound}
    Let $\rho \in \mathcal S(H)$ be a quantum state on $H = \mathbb C^{d_{H}}$.
    Then, there exists a sequence of projective measurements $M_n \succeq 0$ such that, for every state $\sigma \in \mathcal S(H)$ and every $n \in \mathbb N$,
    \begin{equation}
        \mathcal R_{\sigma/\rho}(M_n) = \frac{\Tr(M_n \sigma^{\otimes n})}{\Tr(M_n \rho^{\otimes n})} \leq C(\rho)^{-1}\exp(-(n-d_{H}^{2}) D^{*}(\rho \| \sigma)),
    \end{equation}
    where $C(\rho) \in (0,1]$ is given by 
    \begin{equation}
        C(\rho) \coloneqq x_1^{s} x_2^{s-1} \cdots x_{s-1}^{2} x_{s}.
    \end{equation}
    Here $(x_1 \geq \cdots \geq x_{d_{H}})$ are the eigenvalues of $\rho$, and $s = \mathrm{rank}(\rho)$.
\end{lem}
In the main text, we use \cref{lem:likelihood_ratio_fidelity}, which is a version of \cref{lem:likelihood_ratio_keyl_upper_bound} stated in terms of the fidelity $F(\rho, \sigma)$ rather than the quantity $D^{*}(\rho\|\sigma)$. To justify this replacement, we recall that $D^{*}(\rho\|\sigma)$ coincides with the $\alpha \to 1^-$ limit of the \textit{reverse sandwiched relative entropy} \cite{audenaert2013alpha}
\begin{equation}
    \dsym_{\alpha,1-\alpha}(\rho\|\sigma)
    = \frac{1}{\alpha-1} \log \Tr\left[\left( \rho^{\frac{\alpha}{2(1-\alpha)}} \sigma \rho^{\frac{\alpha}{2(1-\alpha)}} \right)^{1-\alpha} \right].
\end{equation}
Using the monotonicity of $\dsym_{\alpha,1-\alpha}(\rho\|\sigma)$ in $\alpha \in (0,1)$~\cite[Lemma 24]{rubboli2024mixed}, one obtains the lower bound
\begin{equation}
    D^{*}(\rho\|\sigma) \geq \dsym_{\frac12,\frac12}(\rho\|\sigma) = - \log F(\rho, \sigma).
\end{equation}
Equivalently,
\begin{equation}
    \exp(- D^{*}(\rho\|\sigma)) \leq F(\rho, \sigma).
\end{equation}
Thus, the likelihood-ratio bound in \cref{lem:likelihood_ratio_keyl_upper_bound} can be relaxed to the fidelity-based bound stated in \cref{lem:likelihood_ratio_fidelity}.

We now prove \cref{lem:likelihood_ratio_keyl_upper_bound}, following the proof of \cite[Theorem 15, Appendix J]{fraser2022sufficient}, with minor adjustments.
\begin{proof}[Proof of \cref{lem:likelihood_ratio_keyl_upper_bound}]
    The sequence of measurements depends on the projection operators $\Phi_{\lambda}^{U}$, introduced in \cref{eq:coherent_projectors}.
     Considering the ratio of $\Tr(\Phi_{\lambda}^{U} \sigma^{\otimes n})$ and $\Tr(\Phi_{\lambda}^{U} \rho^{\otimes n})$,
    from \cref{eq:gen_pow_func_identity} we see that the $\dim(V_{\lambda})$ terms cancel, and we obtain:
    \begin{equation}
        \frac{\Tr(\Phi_{\lambda}^{U} \sigma^{\otimes n})}{\Tr(\Phi_{\lambda}^{U} \rho^{\otimes n})}
        = \frac{\Delta_{\lambda}(U^{\dagger}\sigma U)}{\Delta_{\lambda}(\mathrm{diag}(x))}.
    \end{equation}
    Our goal is to choose a Young diagram $\lambda$, whose normalized shape approximates the eigenvalues $x$ of $\rho$.
    For this, we fix a positive integer $k$ and consider the Young diagram $\mu$ defined such that
    \begin{equation}
        \label{eq:rounding_differences_up}
        \mu_j - \mu_{j+1} \coloneq \lceil k(x_j - x_{j+1}) \rceil.
    \end{equation}
    Letting $\delta_{j} \coloneqq \mu_j - kx_j$ denote the rounding error, we see that 
    \begin{equation}
        0 \leq \delta_{j} - \delta_{j+1}\leq 1,
    \end{equation}
    and 
    \begin{equation}
        k \leq \abs{\mu} \leq k + d^{2}.
    \end{equation}
    Since $\Delta_{a+b}(X) = \Delta_{a}(X) \Delta_{b}(X)$ for any $X$ and any vectors $a,b$, for $\mu = kx + \delta$ we get the identity
    \begin{equation}
        \frac{\Delta_{\mu}(U^{\dagger}\sigma U)}{\Delta_{\mu}(\mathrm{diag}(x))} 
        =\frac{\Delta_{k x}(U^{\dagger}\sigma U)}{\Delta_{kx}(\mathrm{diag}(x))}\frac{\Delta_{\delta}(U^{\dagger}\sigma U)}{\Delta_{\delta}(\mathrm{diag}(x))}.
    \end{equation}
    The ratio we are interested in, is
    \begin{equation}
        \frac{\Delta_{k x}(U^{\dagger}\sigma U)}{\Delta_{kx}(\mathrm{diag}(x))} = \exp ( - k  D^{*}(\rho\|\sigma)).
    \end{equation}
    The remaining term, depending on $\delta$, can be bounded as follows.
    First, we have $\mathrm{pm}_j(U^{\dagger}\sigma U) \leq 1$ and $\delta_{j} - \delta_{j+1} \geq 0$, yielding
    \begin{equation}
        \Delta_{\delta}(U^{\dagger}\sigma U) = \prod_{j=1}^{d}\mathrm{pm}_j(U^{\dagger}\sigma U)^{\delta_{j} - \delta_{j+1}} \leq 1.
    \end{equation}
    Analogously, $\mathrm{pm}_j(\mathrm{diag}(x)) \leq 1$ and $\delta_{j} - \delta_{j+1} \leq 1$, with $\delta_{j} - \delta_{j+1} = 0$ whenever $x_j - x_{j+1} = 0$.
    Let $s = \mathrm{rank}(\rho)$ such that $x_{s} > 0$ is the smallest non-zero eigenvalue of $\rho$.
    Then, we conclude that the following bound holds:
    \begin{align}
        \Delta_{\delta}(\mathrm{diag}(x)) 
        = \prod_{j=1}^{d}\mathrm{pm}_j(\mathrm{diag}(x))^{\delta_{j} - \delta_{j+1}}
        \geq \prod_{\substack{j=1\\x_j > 0}}^{d}\mathrm{pm}_j(\mathrm{diag}(x))
        = C(\rho).
    \end{align}
    Altogether, we have
    \begin{equation}
        \label{eq:mu_ratio_bound}
        \frac{\Tr(\Phi_{\mu}^{U} \sigma^{\otimes \abs{\mu}})}{\Tr(\Phi_{\mu}^{U} \rho^{\otimes \abs{\mu}})}
        \leq C(\rho)^{-1} \exp ( - k  D^{*}(\rho\|\sigma)).
    \end{equation}
    Next, for any value of $n$, let $k$ be the \textit{largest} positive integer such that $\mu$, defined by \cref{eq:rounding_differences_up}, has size $\abs{\mu} \leq n$.
    This ensures that $k + d^{2} \geq n$, since otherwise $k + d^{2} < n$ and taking $k'= k+1$ yields $\mu'$ with size $\abs{\mu'} \leq k' + d^2 = k + d^2 + 1 < n + 1$, or $\abs{\mu'} \leq n$.
    Finally, define
    \begin{equation}
        M_n \coloneqq \Phi_{\mu}^{U} \otimes I^{\otimes n - \abs{\mu}}.
    \end{equation}
    Since $k \geq n - d^{2}$, from \cref{eq:mu_ratio_bound} we conclude that
    \begin{equation}
        \frac{\Tr(M_n \sigma^{\otimes n})}{\Tr(M_n \rho^{\otimes n})} \leq C(\rho)^{-1} \exp ( - (n-d^2)  D^{*}(\rho\|\sigma)).
    \end{equation}
\end{proof}

\end{appendices}

\end{document}